\documentclass[aos,preprint]{imsart}

\usepackage{macros}

\usepackage{graphicx}

\usepackage{hyperref}
\usepackage{cleveref}
\crefname{theorem}{theorem}{theorems}
\Crefname{theorem}{Theorem}{Theorems}

\crefname{proposition}{proposition}{propositions}
\Crefname{proposition}{Proposition}{Propositions}

\crefname{lemma}{lemma}{lemmas}
\Crefname{lemma}{Lemma}{Lemmas}

\crefname{corollary}{corollary}{corollaries}
\Crefname{corollary}{Corollary}{Corollaries}

\crefname{assumption}{assumption}{assumptions}
\Crefname{assumption}{Assumption}{Assumptions}

\crefname{definition}{definition}{definitions}
\Crefname{definition}{Definition}{Definitions}

\crefname{remark}{remark}{remarks}
\Crefname{remark}{Remark}{Remarks}

\begin{document}

\begin{frontmatter}
\title{Duality and Error for Predictively Oriented Inference}
\runtitle{Duality and Error for Predictively Oriented Inference}

\begin{aug}
\author[A]{\fnms{Aurya}~\snm{Javeed}\ead[label=e1]{asjavee@sandia.gov}}
\author[A]{\fnms{Drew P.}~\snm{Kouri}\ead[label=e2]{dpkouri@sandia.gov}}
\author[A]{\fnms{Teresa}~\snm{Portone}\ead[label=e3]{tporton@sandia.gov}}
\author[A]{\fnms{Rebekah}~\snm{White}\ead[label=e4]{rebwhit@sandia.gov}}

\address[A]{Optimization and Uncertainty Quantification,
Sandia National Laboratories,
Albuquerque, NM, United States of America
\printead[presep={,\ }]{e1,e2,e3,e4}}
\end{aug}

\begin{abstract}

Predictively oriented (PrO) inference quantifies uncertainty by selecting a distribution over model parameters to optimize a scoring rule applied to the induced predictive distribution, together with a divergence penalty from a reference distribution.
By applying the scoring rule after averaging model densities, PrO inference targets predictive performance, accounting for model misspecification.
We focus on the logarithmic score with general $\phi$-divergence regularization.
Our contributions are twofold.
First, we derive a finite-dimensional dual formulation of PrO inference.
For $n$ observations, the dual problem has $n+1$ variables.
We establish zero-duality-gap criteria and optimality conditions that relate the primal and dual solutions.
When primal and dual solutions exist, these conditions yield a semi-analytical representation of the PrO posterior and certificates for assessing the accuracy of numerical solutions.
For Kullback--Leibler regularization, the posterior has an exponential form.
Second, we derive a finite-sample excess predictive-risk bound for approximate PrO posteriors that separates sampling fluctuation, approximation under a divergence budget, regularization, and numerical optimization error.
The result applies even when the benchmark predictive risk is not attained by any probability distribution over the model parameters having finite divergence from the reference distribution.
We use an exactly solvable categorical example to show that predictive-risk convergence can imply convergence to a unique predictive distribution even though the parameter distributions have no weak limit on the original parameter space.
The example also shows that different $\phi$-divergences can require different regularization schedules.
We conclude with a misspecified Gaussian location-mixture example that illustrates the dual computation, primal recovery, and numerical accuracy checks.

\end{abstract}

\begin{keyword}[class=MSC]
\kwdgroup[type=primary]{
\kwd{62F15}
\kwd{90C46}
\kwd{90C15}}
\end{keyword}

\begin{keyword}
\kwd{Optimization}
\kwd{conjugate duality}
\kwd{stochastic programming}
\kwd{posterior predictive distributions}
\kwd{misspecification}
\kwd{generalized Bayes}
\end{keyword}

\end{frontmatter}

\section{Introduction}
\label{sec:intro}

Probabilistic inference provides a framework for quantifying uncertainty and forming predictions from observational data.
The predictive performance of such methods can deteriorate under model misspecification, that is, when the statistical model does not contain the data-generating distribution.
Methods such as generalized Bayesian inference and Gibbs posteriors \cite{BissiriHolmesWalker2016,Zhang2006b,zhang_e-entropy_2007, jiang_gibbs_2008} have been developed in part to improve robustness to misspecification.
These methods nevertheless remain parameter-oriented: they evaluate losses associated with individual model components and then average those losses under a probability distribution over the parameters.
Under misspecification, this focus on model components can lead to predictive distributions that do not adequately represent uncertainty \cite{bjk_kleijn_bernstein-von-mises_2012}.

Predictively oriented (PrO) inference \cite{McLatchieCheriefAbdellatifFrazierKnoblauch2025} has been proposed as a framework that directly targets prediction.
PrO inference formulates posterior construction as an infinite-dimensional regularized optimization problem over probability distributions on the model parameters.
The objective combines a loss term, which measures predictive agreement with the observed data through a statistical scoring rule, and a divergence term that regularizes the candidate posterior by discouraging departure from a reference (i.e., prior) distribution.
The key distinction between PrO and Bayesian or generalized Bayesian formulations is the order in which averaging and scoring occur.
Let \(\Theta\) denote the parameter space, \(p_\param\) the model density for \(\param\in\Theta\), \(Q\) a probability distribution on \(\Theta\), and \(x\) a point in the sample space.
In the logarithmic case considered here, the distinction is
\[
  \int_\Theta -\log p_\param(x)\,dQ(\param)
  \qquad\text{versus}\qquad
  -\log\left\{ \int_\Theta p_\param(x)\,dQ(\param) \right\}.
\]
The first expression averages the logarithmic losses of the individual model components, whereas the second applies the logarithmic loss after averaging the component densities and therefore scores the resulting predictive distribution directly.

Although conceptually appealing in its alignment with the predictive goals of many inference problems, the formulation of PrO inference as an infinite-dimensional optimization problem presents challenges for computation and theoretical analysis.
Existing approaches have largely treated this problem through variational or particle-based techniques \cite{McLatchieCheriefAbdellatifFrazierKnoblauch2025}.
We instead use tools from convex analysis to derive a finite-dimensional dual formulation of the PrO problem for the logarithmic score and general $\phi$-divergence regularization.
The resulting dual problem has only $n+1$ scalar variables, where $n$ is the number of observations, even when the primal problem is infinite dimensional.
We provide conditions under which the primal and dual problems have the same optimal value and the dual optimum is attained.
We also derive optimality conditions that relate the primal and dual solutions.
When a primal solution exists, these conditions recover its density relative to the prior, and hence its predictive distribution, from the dual variables and also provide certificates for assessing the accuracy of numerical solutions.
For Kullback--Leibler (KL) regularization, the optimal density relative to the prior has an explicit exponential form.

Our second primary contribution is an analysis of predictive-risk error for approximate PrO posteriors.
We take as a benchmark the infimum of the predictive risk over admissible probability distributions having finite divergence.
This infimum need not be attained and may be approachable only along a sequence of distributions whose divergences increase without bound.
To account for this possibility, we introduce an approximation profile that records the smallest excess predictive risk achievable under a given divergence budget.
We derive a finite-sample excess-risk bound that separates finite-sample fluctuations, approximation under the divergence budget, regularization, and inexact numerical optimization.
The bound yields predictive-risk consistency when these effects vanish and convergence rates when their magnitudes can be quantified.
This analysis is complementary to the work of \cite{McLatchieCheriefAbdellatifFrazierKnoblauch2025}, which establishes broad predictive guarantees for PrO posteriors and compares their behavior with Bayesian and Gibbs methods.

To make our theoretical results concrete, we consider two examples.
The first is an exactly solvable categorical example that demonstrates the importance of allowing nonattainment.
In the example, no probability distribution over the original parameter space attains the benchmark predictive risk.
Nevertheless, any sequence of such distributions whose predictive risks approach the benchmark induces predictive distributions that converge in total variation to a unique limit.
The probability distributions over the parameters themselves do not converge weakly on the original parameter space; instead, their mass moves toward the excluded boundary of that space.
For this example, our conclusion goes beyond the general predictive-performance guarantee of \cite{McLatchieCheriefAbdellatifFrazierKnoblauch2025}: it identifies the limiting predictive distribution and characterizes the behavior of the distributions over the parameters even though no optimal parameter distribution exists.
We also obtain sharp approximation rates for KL and $\chi^2$ divergences, showing that the two penalties require different regularization schedules.
The second example considers a misspecified Gaussian location-mixture model. We first specialize the predictive-risk bound to this setting and then illustrate the KL dual computation, primal recovery, and numerical optimality checks.

The remainder of the paper is organized as follows. \Cref{sec:background} introduces the relevant notation, the statistical model, predictive distributions, and the variational formulation of Bayesian inference. \Cref{sec:problem-formulation} defines the $\phi$-divergence regularization, formalizes the empirical PrO problem, and states the assumptions used in our analysis. \Cref{sec:finite-sample-duality} derives the dual PrO problem, establishes stability and optimality conditions, and specializes the results to KL regularization. \Cref{sec:predictive-convergence-rates} develops the predictive-risk analysis. \Cref{sec:examples} presents the categorical and Gaussian location-mixture examples. Finally, \Cref{sec:discussion-conclusion} concludes.

\section{Background and notation}
\label{sec:background}

This section introduces the data, statistical model, predictive distributions, and notation used throughout the paper.
We also briefly review the variational formulation of Bayesian inference to provide context for the PrO problem in the next section.

\subsection{Data and statistical model}

Let $(\mathsf X,\mathcal X)$ be a measurable sample space. We observe
\[
  X_1,\ldots,X_n
  \overset{\mathrm{iid}}{\sim}
  \datadist,
\]
where $\datadist$ is the data-generating distribution on
$(\mathsf X,\mathcal X)$.
For a probability measure $P$ on $(\mathsf X,\mathcal X)$ and a
$P$-integrable function $\psi$, we write
\[
  P\psi\coloneqq\int_{\mathsf X}\psi\,dP.
\]
The empirical measure is
$P_n\coloneqq n^{-1}\sum_{i=1}^n\delta_{X_i}$, where $\delta_x$ denotes the
Dirac probability measure at $x$; hence
\[
  P_n\psi=\frac{1}{n}\sum_{i=1}^n\psi(X_i).
\]
To approximate the data-generating distribution $\datadist$, let $(\Theta,\mathcal T)$ be a measurable parameter space, and let
\[
  \{P_\param:\param\in\Theta\}
\]
be a statistical model on $(\mathsf X,\mathcal X)$. Let $\priordist$ be a
prior probability measure on $(\Theta,\mathcal T)$.
Assume that
\[
  P_\param\ll\mu,
  \qquad \param\in\Theta,
\]
where $\mu$ is a common $\sigma$-finite measure on
$(\mathsf X,\mathcal X)$ and $\ll$ denotes absolute continuity. We write
\[
  p_\param
  \coloneqq
  \frac{dP_\param}{d\mu},
\]
and assume that the map
$(x,\param)\mapsto p_\param(x)$ is jointly measurable.

\subsection{Predictive distributions}

We call a probability measure $Q$ on $(\Theta,\mathcal T)$ a mixing
distribution. The predictive distribution induced by $Q$ is the probability
measure $P_Q$ on $(\mathsf X,\mathcal X)$ defined by
\begin{equation*}
\label{eq:predictive-mixture-distribution}
  P_Q(A)
  \coloneqq
  \int_\Theta P_\param(A)\,dQ(\param),
  \qquad A\in\mathcal X.
\end{equation*}
Tonelli's theorem implies that $P_Q\ll\mu$. A density of $P_Q$ with respect to $\mu$ is
\begin{equation*}
\label{eq:predictive-mixture-density}
  p_Q(x)
  \coloneqq
  \int_\Theta p_\param(x)\,dQ(\param),
  \qquad x\in\mathsf X.
\end{equation*}

\subsection{Bayesian inference}

For comparison with the PrO problem presented in the next section, we review the variational formulation of Bayesian inference.
For a probability measure $Q$ on $(\Theta,\mathcal T)$, the Kullback--Leibler divergence of $Q$ from $\priordist$ is
\[
  D_{\mathrm{KL}}(Q\|\priordist)
  \coloneqq
  \begin{cases}
  \displaystyle
  \int_\Theta
  \log\left(\frac{dQ}{d\priordist}\right)\,dQ,
  & Q\ll\priordist,\\[0.8em]
  +\infty,
  & \text{otherwise}.
  \end{cases}
\]
Under standard regularity conditions, the Bayesian posterior minimizes
\begin{equation}
\label{eq:bayes-variational}
  -\int_\Theta
  \left\{
  \frac1n\sum_{i=1}^n\log p_\param(X_i)
  \right\}
  dQ(\param)
  +
  \frac1nD_{\mathrm{KL}}(Q\|\priordist)
\end{equation}
over probability measures $Q$ on $(\Theta,\mathcal T)$ \cite{Catoni2007,DonskerVaradhan1975}.
The first term in \eqref{eq:bayes-variational} averages the logarithmic losses of the individual model distributions under $Q$.
Generalized Bayesian methods retain this variational structure while replacing the negative log-likelihood with a general loss \cite{BissiriHolmesWalker2016}.
In contrast, PrO inference applies the logarithmic loss directly to the predictive distribution $P_Q$, as formalized in \Cref{sec:problem-formulation}.

\section{Problem formulation}
\label{sec:problem-formulation}

This section specifies the admissible mixing distributions and divergence regularization, formulates the empirical PrO problem, and states the assumptions used in the duality analysis.

\subsection{Admissible mixing distributions and divergence}

Let $\mathcal E$ be a vector space of real-valued, $\mathcal T$-measurable functions, identified up to
$\priordist$-almost-sure equality.
The additional structure imposed on $\mathcal E$ is stated in \Cref{ass:finite-sample-convex-setup}.
Define the admissible class of mixing distributions by
\begin{equation}
\label{eq:admissible-mixing-class}
  \mathcal Q_{\mathcal E}
  \coloneqq
  \left\{
      Q:
      \;
      Q\text{ is a probability measure on }(\Theta,\mathcal T),
      \;\;
      Q\ll\priordist,
      \;\;
      \frac{dQ}{d\priordist}\in\mathcal E
  \right\}.
\end{equation}
Under the identification $q=dQ/d\priordist$, the corresponding set of admissible mixing densities is
\begin{equation*}
\label{eq:primal-feasible-set}
  \mathcal C
  \coloneqq
  \left\{
      q\in\mathcal E:
      \;
      \int_\Theta q\, d\priordist=1,
      \;\;
      q\geq0\;\priordist\text{-almost surely}
  \right\}.
\end{equation*}
Conversely, every $q\in\mathcal C$ defines a mixing distribution $Q\in\mathcal Q_{\mathcal E}$ by
\[
  Q(A)
  \coloneqq
  \int_A q\,d\priordist,
  \qquad A\in\mathcal T.
\]
Let $\phi:\mathbb R\to[0,+\infty]$ be a lower-semicontinuous convex function satisfying
\begin{equation*}
\label{eq:phi-assumptions}
  \phi(1)=0
  \qquad\text{and}\qquad
  \phi(t)=+\infty
  \quad\text{for}\quad
  t<0.
\end{equation*}
For a probability measure $Q$ on $(\Theta,\mathcal T)$, define its $\phi$-divergence from $\priordist$ by
\begin{equation*}
\label{eq:phi-divergence}
  D_\phi(Q\|\priordist)
  \coloneqq
  \begin{cases}
  \displaystyle
  \int_\Theta
  \phi\left(
      \frac{dQ}{d\priordist}(\param)
  \right)
  d\priordist(\param),
  & Q\ll\priordist,\\[1em]
  +\infty,
  & \text{otherwise}.
  \end{cases}
\end{equation*}
The associated integral functional $\Phi:\mathcal E\to[0,+\infty]$ is
\begin{equation*}
\label{eq:Phi-definition}
  \Phi(q)
  \coloneqq
  \int_\Theta
  \phi(q(\param))\,d\priordist(\param).
\end{equation*}
Thus, if $Q\in\mathcal Q_{\mathcal E}$ and $q=dQ/d\priordist$, then
\[
  D_\phi(Q\|\priordist)
  =
  \Phi(q).
\]
The convention $\phi(t)=+\infty$ for $t<0$ encodes the nonnegativity of a mixing density in the extended real-valued functional $\Phi$.
Finally, define the Fenchel conjugate of $\phi$ by
\begin{equation*}
\label{eq:phi-conjugate}
  \phi^*(s)
  \coloneqq
  \sup_{t\in\mathbb R}
  \{st-\phi(t)\}
  =
  \sup_{t\geq0}
  \{st-\phi(t)\},
  \qquad s\in\mathbb R.
\end{equation*}
The corresponding conjugacy relation for the integral functional $\Phi$ is recorded in \Cref{rem:integral-conjugacy}.

\subsection{PrO inference}

Given the observations $X_1,\ldots,X_n$, define
\begin{equation*}
\label{eq:fi-definition}
  f_i(\param)
  \coloneqq
  p_\param(X_i),
  \qquad i=1,\ldots,n.
\end{equation*}
We regard each $f_i$ as a fixed measurable function on $\Theta$.
For a mixing distribution $Q$, define the predictive log loss by
\begin{equation*}
\label{eq:loss-Q-definition}
  \ell_Q(x)
  \coloneqq
  -\log p_Q(x),
\end{equation*}
with the convention that $\ell_Q(x)=+\infty$ when $p_Q(x)=0$. If
$Q\ll\priordist$ and $q=dQ/d\priordist$, then
\begin{equation*}
\label{eq:empirical-predictive-loss}
  P_n\ell_Q
  =-\frac1n\sum_{i=1}^n\log p_Q(X_i)
  =-\frac1n\sum_{i=1}^n
  \log\left\{
  \int_\Theta f_i(\param)q(\param)\,d\priordist(\param)
  \right\}.
\end{equation*}
Unlike the data-fit term in \eqref{eq:bayes-variational}, this expression averages the model densities before applying the logarithmic loss.
For a regularization parameter $\varepsilon_n>0$, the PrO problem has the equivalent measure-space and density-space formulations
\begin{equation}
\label{eq:primal-normalized}
\tag{$\primeProb$}
  \inf_{Q\in\mathcal Q_{\mathcal E}}
  \left\{
  P_n\ell_Q + \varepsilon_nD_\phi(Q\|\priordist)
  \right\}
  =
  \inf_{q\in\mathcal C}
  \left[
  -\frac1n\sum_{i=1}^n
  \log
  \int_\Theta f_i(\param)q(\param)\,d\priordist(\param)
  + \varepsilon_n\Phi(q)
  \right].
\end{equation}
We refer to the common optimization problem in \eqref{eq:primal-normalized} as the primal PrO problem, or simply the primal problem, and let $\operatorname{val}(\primeProb)$ denote its optimal value.
Since $Q\mapsto p_Q(X_i)$ is linear, $z\mapsto-\log z$ is convex on $(0,+\infty)$, and $D_\phi(\,\cdot\,\|\priordist)$ is convex, the objective in \eqref{eq:primal-normalized} is convex in $Q$.
When the infimum in \eqref{eq:primal-normalized} is attained, let $Q_n^\star\in\mathcal Q_{\mathcal E}$ denote an exact PrO solution and write
\(
  q_n^\star
  \coloneqq
  \frac{dQ_n^\star}{d\priordist}.
\)
More generally, for an optimization tolerance $\xi_n\geq0$, a $\xi_n$-approximate PrO solution is any $\widehat Q_n\in\mathcal Q_{\mathcal E}$ satisfying
\begin{equation}
\label{eq:approximate-primal-optimality}
  P_n\ell_{\widehat Q_n}
  +
  \varepsilon_nD_\phi(\widehat Q_n\|\priordist)
  \leq
  \text{val}(\primeProb)+\xi_n.
\end{equation}
We write
\(
  \widehat q_n
  \coloneqq
  \frac{d\widehat Q_n}{d\priordist}.
\)
When \(\xi_n=0\), $\widehat Q_n$ is an exact PrO solution.

\subsection{Standing assumptions}

The dual PrO problem is formulated in a paired function space setting.
A vector space $V$ of real-valued, \(\mathcal T\)-measurable functions is called \emph{decomposable} in association with $\priordist$ if ($i$) it contains all bounded real-valued, \(\mathcal T\)-measurable functions and ($ii$) for all $u\in V$ and $E\in\mathcal T$, the function $\chi_E u\in V$, where $\chi_E$ denotes the characteristic function of $E$.
Since $\priordist$ is a finite measure, this definition of decomposability coincides with that in \cite{Rockafellar1976Integral}.

\begin{assumption}[Problem assumptions]
\label[assumption]{ass:finite-sample-convex-setup}
For each realized sample $X_1,\ldots,X_n$, the following conditions hold.
\begin{enumerate}
  \item The spaces $\mathcal E$ and $\mathcal E^*$ are decomposable vector spaces of real-valued, $\mathcal T$-measurable functions
  that are paired with respect to the bilinear form
  \[
  \langle q, q^*\rangle = \int_\Theta q(\theta)q^*(\theta)\,d\priordist(\theta), \quad q\in\mathcal{E},\quad q^*\in\mathcal{E}^*,
  \]
  which means \(\mathcal E\) and \(\mathcal E^*\) have locally convex topologies that are compatible with the bilinear form \cite{Rockafellar1974}.
  The space \(\mathcal{E}^*\) need not be the entirety of the topological dual space of \(\mathcal E\).

  \item  The functions $f_1,\ldots,\,f_n\in \mathcal{E}^*$.

  \item There exists \(q^\circ\in\mathcal C\cap\dom\Phi\) such that
  \[
    \int_\Theta f_i(\theta)q^\circ(\theta)\,d\priordist(\theta)>0,
    \qquad i=1,\ldots,n.
  \]
\end{enumerate}
\end{assumption}

The first two conditions imply that the maps
\(
  q
  \mapsto
  \int_\Theta q\,d\priordist
\)
and
\[
  q
  \mapsto
  \int_\Theta f_iq\,d\priordist,
  \qquad i=1,\ldots,n,
\]
are continuous linear functionals on $\mathcal E$.
The third condition ensures that \eqref{eq:primal-normalized} is feasible.

\begin{remark}[Fenchel conjugate]
\label[remark]{rem:integral-conjugacy}

Let $\sigma(\mathcal E,\mathcal E^*)$ denote the weak topology on $\mathcal E$ induced by the pairing $\langle\cdot,\cdot\rangle$; that is, the coarsest topology for which $q\mapsto\langle q,q^*\rangle$ is continuous for every $q^*\in\mathcal E^*$.
The integral functionals generated by $\phi$ and $\phi^*$ are proper.
Indeed, decomposability of $\mathcal E$ implies that $1\in\mathcal E$, and
\[
  \Phi(1)
  =
  \int_\Theta \phi(1)\,d\priordist
  =
  0.
\]
Moreover, since $\phi\geq0$ and $\phi(1)=0$, we have $\phi^*(0)=-\inf_{t\in\mathbb R}\phi(t)=0$.
Because
$0\in\mathcal E^*$,
\[
  \int_\Theta \phi^*(0)\,d\priordist
  =
  0.
\]
The functional $\Phi$ is convex because $\phi$ is convex.
By \cite[Theorem~3C]{Rockafellar1976Integral}, decomposability of $\mathcal E$ and $\mathcal E^*$ implies that $\Phi$ is $\sigma(\mathcal E,\mathcal E^*)$-lower semicontinuous and that
\begin{equation}
\label{eq:integral-conjugacy}
  \Phi^*(q^*)
  \coloneqq
  \sup_{q\in\mathcal E}
  \left\{
      \langle q,q^*\rangle-\Phi(q)
  \right\}
  =
  \int_\Theta
  \phi^*(q^*(\theta))\,d\priordist(\theta),
  \qquad q^*\in\mathcal E^*.
\end{equation}
In particular, \eqref{eq:integral-conjugacy} applies to every $q^*\in\operatorname{span}\{1,f_1,\ldots,f_n\}$.
\end{remark}

\begin{remark}[Examples of spaces]
\label{rem:examples-primal-space}
Suppose that $\Theta=\{\param_1,\ldots,\param_K\}$ and $\priordist_j\coloneqq\priordist(\{\param_j\})>0$.
We may take $\mathcal E=\mathcal E^*=\mathbb R^K$ with the weighted pairing
\(
  \langle q,s\rangle
  =
  \sum_{j=1}^K\priordist_jq_js_j.
\)
Then
\[
  \mathcal C
  =
  \left\{
  q\in\mathbb R^K:
  \;\;
  \sum_{j=1}^K\priordist_jq_j=1
  \;\;
  q_j\geq0
  \right\}
\]
with
\[
  \Phi(q)
  =
  \sum_{j=1}^K\priordist_j\phi(q_j)
  \quad
  \text{and}
  \quad
  \Phi^*(s)
  =
  \sum_{j=1}^K\priordist_j\phi^*(s_j).
\]
The density $q\equiv1$, corresponding to $Q=\priordist$, is feasible and has zero divergence.
If, in addition,
\[
  p_\priordist(X_i)
  =
  \sum_{j=1}^K\priordist_jf_i(\param_j)
  >0,
  \qquad i=1,\ldots,n,
\]
then $q\equiv1$ satisfies the third condition in
\Cref{ass:finite-sample-convex-setup}.

For infinite-dimensional models, typical choices include paired Lebesgue spaces $L^p(\priordist)$ and $L^{p'}(\priordist)$, as well as complementary Orlicz spaces.
The appropriate choice depends on the growth of $\phi$ and the integrability of the functions $f_i$.
\end{remark}

\section{Duality}
\label{sec:finite-sample-duality}

\Cref{sec:dual_derv} establishes the existence of the dual formulation of the density-based PrO posterior optimization problem, and \Cref{sec:stability} provides conditions under which this formulation is stable, that is, when ($i$) the optimal values of the primal and dual problems coincide and ($ii$) the dual problem has a solution.
The latter section then establishes necessary optimality conditions that relate the solutions of the primal problem and the dual problem to obtain a novel semi-analytical form for the PrO posterior along with certificates that verify the accuracy of approximate solutions of the dual problem.
\Cref{sec:kl_regularization} shows that in the case of the KL regularization, the optimality conditions yield a posterior that has an exponential form.
\Cref{sec:practical_stability} provides verifiable sufficient conditions under which stability is achieved and discusses these conditions in the context of some specific problem classes.

\subsection{Deriving the dual problem}\label{sec:dual_derv}

The dual problem is derived from the PrO problem in the standard way: the PrO problem is embedded into a parameterized family of optimization problems with the dual problem derived from the optimal-value function of that family.

For notational convenience, let
\begin{equation*}
\label{eq:T-definition}
  T_n:\mathcal E\rightarrow\mathbb R^n,
  \qquad
  (T_nq)_i
  \coloneqq
  \int_\Theta f_i(\param)q(\param)\,d\priordist(\param),
  \qquad i=1,\ldots,n.
\end{equation*}
By \Cref{ass:finite-sample-convex-setup}, each coordinate of \(T_n\) is a continuous linear functional on \(\mathcal E\), and hence \(T_n\) is continuous.
Define the extended real-valued convex function
\begin{equation*}
\label{eq:g-definition}
  g(z)
  \coloneqq
  \begin{cases}
  -\dfrac1n\displaystyle\sum_{i=1}^n\log z_i,
      & z\in\mathbb R^n_{++},\\[0.75em]
  +\infty,
      & z\notin\mathbb R^n_{++}.
  \end{cases}
\end{equation*}
For \(Q\ll\priordist\) with density \(q=dQ/d\priordist\), we have
\(
  (T_nq)_i=p_Q(X_i)
\), \(i=1,\ldots,n\),
and therefore
\[
  g(T_nq)=P_n\ell_Q.
\]
It follows that the primal PrO problem \eqref{eq:primal-normalized} can be written as
\begin{equation*}
  \label{eq:finite-sample-primal}
  \inf_{q\in\mathcal E}
  \left\{
  g(T_nq)+\varepsilon_n\Phi(q):
  \int_\Theta q(\param)\,d\priordist(\param)=1
  \right\},
\end{equation*}
where the nonnegativity constraint is encoded by \(\Phi\), since \(\phi(t)=+\infty\) for \(t<0\).
We define a family of optimization problems parameterized by \((u,r)\in\mathbb R^n\times\mathbb R\) through the optimal-value function
\begin{equation*}
\label{eq:finite-sample-perturbation-value}
  \valfun(u,r)
  \coloneqq
  \inf_{q\in\mathcal E}
  \left\{
  g(T_nq+u)+\varepsilon_n\Phi(q):
  \int_\Theta q(\param)\,d\priordist(\param)=1+r
  \right\}.
\end{equation*}
Note that \(\valfun(0,0)\) is the primal problem \eqref{eq:primal-normalized}.
To state the dual problem compactly, for \(\alpha\in\mathbb R^n_{++}\coloneqq (0,+\infty)^n\), define the weighted average
\begin{equation*}
\label{eq:h-n-alpha}
  \avLik{\alpha}
  \coloneqq
  \frac1n\sum_{i=1}^n\alpha_i f_i(\param).
\end{equation*}

\begin{proposition}[Dual formulation]
\label[proposition]{prop:empirical-dual-representation}
Suppose Assumption~\ref{ass:finite-sample-convex-setup} holds.
The dual problem associated with the primal problem \eqref{eq:primal-normalized} is
\begin{equation}
\label{eq:finite-sample-dual}
\tag{$\dualProb$}
  \sup_{\substack{\alpha\in\mathbb R^n_{++}\\ \eta\in\mathbb R}}
  \left\{
  \frac1n\sum_{i=1}^n(1+\log\alpha_i)
  -\eta
  -\varepsilon_n
  \int_\Theta
  \phi^*\!\left(
  \frac{\avLik{\alpha}-\eta}{\varepsilon_n}
  \right)
  d\priordist(\param)
  \right\}.
\end{equation}
\end{proposition}

\begin{proof}
Note that
\(
  \valfun(u,r)=\inf_{q\in\mathcal E}\mathcal F_n(q,u,r),
\)
where
\[
  \mathcal F_n(q,u,r)
  \coloneqq
  \begin{cases}
  g(T_nq+u)+\varepsilon_n\Phi(q),
  &
  \displaystyle \int_\Theta q\,d\priordist=1+r,
  \\
  +\infty,
  &
  \text{otherwise}.
  \end{cases}
\]
For \(\alpha\in\mathbb R^n_{++}\) and \(\eta\in\mathbb R\), the Lagrangian associated with $\mathcal F_n$ \cite{Rockafellar1974} is
\begin{align*}
  L_n(q,\alpha,\eta)
  &\coloneqq
  \inf_{\substack{u\in\mathbb R^n\\ r\in\mathbb R}}
  \left\{
  \mathcal F_n(q,u,r)
  +\frac1n\sum_{i=1}^n\alpha_i u_i
  +\eta r
  \right\}\\
  &=
  \inf_{u\in\mathbb R^n}
  \left\{
  g(T_nq+u) + \varepsilon_n\Phi(q)
  +\frac1n\sum_{i=1}^n\alpha_i u_i
  +\eta \left(\int_\Theta q\,d\priordist-1\right)
  \right\}.
  \end{align*}
Setting \(z=T_nq+u\), so that \(u=z-T_nq\), gives
\begin{align*}
  &\inf_{u\in\mathbb R^n}
  \left\{
  g(T_nq+u)
  +\frac1n\sum_{i=1}^n\alpha_i u_i
  \right\}
  =
  -\frac1n\sum_{i=1}^n\alpha_i(T_nq)_i
  +
  \inf_{z\in\mathbb R^n}
  \left\{
  g(z)+\frac1n\sum_{i=1}^n\alpha_i z_i
  \right\}.
\end{align*}
Moreover, since \(g(z)=+\infty\) outside \(\mathbb R^n_{++}\),
\begin{align*}
  \inf_{z\in\mathbb R^n}
  \left\{
  g(z)+\frac1n\sum_{i=1}^n\alpha_i z_i
  \right\}
  &=
  \frac1n\sum_{i=1}^n
  \inf_{z_i>0}
  \left\{
  -\log z_i+\alpha_i z_i
  \right\}
  =
  \frac1n\sum_{i=1}^n(1+\log\alpha_i),
\end{align*}
where the last equality follows from
\begin{equation}
\label{eq:negative-log-conjugacy}
  -(1 + \log \alpha_i)
  =
  \sup_{z_i>0}
  \{-\alpha_i z_i+\log z_i\},
  \qquad
  \alpha_i>0.
\end{equation}
Therefore,
\begin{align*}
  L_n(q,\alpha,\eta)
  ={}&
  -\int_\Theta \avLikPFree{\alpha}q\,d\priordist
  +\frac1n\sum_{i=1}^n(1+\log\alpha_i)
  +\varepsilon_n\Phi(q)
  +\eta\left(\int_\Theta q\,d\priordist-1\right).
\end{align*}
The dual problem associated with \eqref{eq:primal-normalized} is
\[
  \sup_{\substack{\alpha\in\mathbb R^n_{++}\\ \eta\in\mathbb R}}
  \inf_{q\in\mathcal E}L_n(q,\alpha,\eta).
\]
For fixed \((\alpha,\eta)\),
\begin{align*}
  \inf_{q\in\mathcal E}L_n(q,\alpha,\eta)
  ={}&
  \frac1n\sum_{i=1}^n(1+\log\alpha_i)
  -\eta
  +\inf_{q\in\mathcal E}
  \left\{
  \varepsilon_n\Phi(q)
  -\int_\Theta
  \bigl(\avLikPFree{\alpha}-\eta\bigr)q\,d\priordist
  \right\} \\
  ={}&
  \frac1n\sum_{i=1}^n(1+\log\alpha_i)
  -\eta
  -\varepsilon_n
  \Phi^*\left(
  \frac{\avLikPFree{\alpha}-\eta}{\varepsilon_n}
  \right).
\end{align*}
Finally, the integral-conjugate representation in
\Cref{rem:integral-conjugacy} gives
\[
  \Phi^*\left(
  \frac{\avLikPFree{\alpha}-\eta}{\varepsilon_n}
  \right)
  =
  \int_\Theta
  \phi^*\left(
  \frac{\avLik{\alpha}-\eta}{\varepsilon_n}
  \right)
  \,d\priordist(\param),
\]
which yields \eqref{eq:finite-sample-dual}.
\end{proof}

\subsection{Stability and optimality conditions}\label{sec:stability}

Our next result characterizes the absence of a duality gap, i.e., of a discrepancy between the optimal values of the primal and dual problems, by the behavior of the value function $\valfun$ at the origin.

\begin{theorem}[Duality and stability]
\label{prop:stable-strong-duality}
Suppose Assumption~\ref{ass:finite-sample-convex-setup} holds and \(\valfun(0,0)\) is finite.
Then there is no duality gap between \eqref{eq:primal-normalized} and \eqref{eq:finite-sample-dual} if and only if the value function $\valfun$ is lower semicontinuous at $(0,0)$.
More specifically, if $\valfun$ is finite and continuous in a neighborhood of $(0,0)$, then the PrO problem is stable: there is no duality gap and the supremum in the dual problem is attained.
\end{theorem}

\begin{proof}
A convex value function that is finite and lower semicontinuous at the origin is equivalent to equality of the optimal values of the primal and dual optimization problems, with a finite common value \cite[Chapter~III, Proposition 2.1]{EkelandTemam1999}.
Suppose further that \(\valfun\) is finite and continuous in a neighborhood of \((0,0)\).
Then \(\valfun\) is subdifferentiable at \((0,0)\); hence there is no duality gap and the supremum in \eqref{eq:finite-sample-dual} is attained \cite[Chapter~III, Proposition~2.2]{EkelandTemam1999}.
\end{proof}

The optimality conditions below are obtained from the equality
case of the Fenchel--Young inequality. We therefore recall the
subdifferential, which characterizes when this equality holds \cite{RockafellarWets1998}. For a proper
convex function
\(\psi:\mathbb R\to(-\infty,+\infty]\), define
\begin{equation*}
\label{eq:scalar-subdifferential}
  \partial\psi(s)
  \coloneqq
  \left\{
  t\in\mathbb R:
  \psi(v)\geq\psi(s)+t(v-s)
  \ \text{for every }v\in\mathbb R
  \right\}.
\end{equation*}

\begin{corollary}[Optimality conditions]
\label[corollary]{cor:kkt-system}
Suppose that \Cref{ass:finite-sample-convex-setup} holds, that there is no duality gap, and that the primal and dual optima are attained at \(Q_n^\star\) and \((\alpha^\star,\eta^\star)\), respectively.
Let \( q_n^\star = dQ_n^\star/d\priordist \).
Then
\begin{align}
  q_n^\star(\param)
  &\in
  \partial\phi^*\left(
  \frac{
  \avLik{\alpha^\star}-\eta^\star
  }{
  \varepsilon_n
  }
  \right),
  \qquad
  \priordist\text{-almost surely},
  \label{eq:kkt-q}
  \\[0.5em]
  1 &= \int_\Theta q_n^\star(\param)\,d\priordist(\param),
  \label{eq:kkt-normalization}
  \\[0.5em]
  \alpha_{i}^\star
  &=
  \left\{
  \int_\Theta
  f_i(\param)q_n^\star(\param)\,d\priordist(\param)
  \right\}^{-1}
  =
  \frac{1}{p_{Q_n^\star}(X_i)},
  \qquad
  i=1,\ldots,n.
\label{eq:kkt-alpha}
\end{align}
If \(\phi^*\) is differentiable at the relevant points, then
\begin{equation}
\label{eq:kkt-q-differentiable}
  q_n^\star(\param)
  =
  (\phi^*)'\left(
  \frac{
  \avLik{\alpha^\star}-\eta^\star
  }{
  \varepsilon_n
  }
  \right),
  \qquad
  \priordist\text{-almost surely}.
\end{equation}
Note that, for approximate dual variables, \eqref{eq:kkt-q} characterizes candidate primal solutions, while \eqref{eq:kkt-normalization} and \eqref{eq:kkt-alpha} provide certificates for assessing approximate satisfaction of the primal--dual optimality conditions.
\end{corollary}

\begin{proof}
By strong duality, the primal and dual optimal values coincide at \(q_n^\star\) and \((\alpha^\star,\eta^\star)\).
Rearranging this equality gives
\[
  \Phi(q_n^\star)
  +
  \Phi^*\left(
  \frac{\avLikPFree{\alpha^\star}-\eta^\star}{\varepsilon_n}
  \right)
  =
  \int_\Theta
  \frac{\avLik{\alpha^\star}-\eta^\star}
  {\varepsilon_n}
  q_n^\star(\param)\,d\priordist(\param).
\]
This is the Fenchel--Young inequality for \(\Phi\) holding with equality, which is possible if and only if \eqref{eq:kkt-q} holds.
Primal feasibility implies \eqref{eq:kkt-normalization}, and equality in \eqref{eq:negative-log-conjugacy} holds only if \(\alpha=1/z\), which implies \eqref{eq:kkt-alpha}.
If \(\phi^*\) is differentiable at the relevant points, then the subdifferential inclusion reduces to \eqref{eq:kkt-q-differentiable}.
\end{proof}

\begin{remark}[Primal attainment]
Stability gives dual attainment and zero duality gap; it does not by itself give existence of a primal minimizer.
Primal attainment follows separately from lower semicontinuity and inf-compactness of
\[
  q\mapsto g(T_nq)+\varepsilon_n\Phi(q)
\]
on the affine set $\{q:\int q\,d\priordist=1\}$.
See \cite[Theorem 2.6]{BonnansShapiro2000}.
This condition is automatic in finite mixtures, where the feasible set is a simplex.
In infinite-dimensional settings, norm compactness is generally too strong, and one typically works with a weaker topology.
\end{remark}

\subsection{KL regularization}
\label{sec:kl_regularization}

We now specialize the general dual formulation and optimality conditions to KL regularization. In this case, the normalization multiplier can be eliminated analytically and the resulting PrO posterior has an exponential form.

\begin{corollary}[KL regularization]
\label[corollary]{cor:kl-gibbs-form}
Suppose that Assumption~\ref{ass:finite-sample-convex-setup} holds, and let
\[
  \phi(t)
  =
  \begin{cases}
  t\log t-t+1, & t\ge0,\\
  +\infty, & t<0,
  \end{cases}
\]
where $0\log 0=0$.
Then
\(
  \phi^*(s)=e^s-1\),
\(
  s\in\mathbb R.
\)
Consequently, the dual problem in \Cref{prop:empirical-dual-representation} becomes
\begin{equation}
\label{eq:kl-dual-with-eta}
  \sup_{\substack{
  \alpha\in\mathbb R^n_{++}\\
  \eta\in\mathbb R
  }}
  \left\{
  \frac1n\sum_{i=1}^n(1+\log\alpha_i)
  -\eta
  -\varepsilon_n
  \int_\Theta
  \left[
  \exp\left(
  \frac{\avLik{\alpha}-\eta}{\varepsilon_n}
  \right)
  -1
  \right]
  \,d\priordist(\param)
  \right\}.
\end{equation}
For fixed $\alpha$, define
\[
  Z_n(\alpha)
  \coloneqq
  \int_\Theta
  \exp\left(
  \frac{\avLik{\alpha}}{\varepsilon_n}
  \right)
  \,d\priordist(\param).
\]
Whenever $Z_n(\alpha)<\infty$, the maximizing value of $\eta$ is
\begin{equation}
\label{eq:kl-optimal-eta}
  \eta(\alpha)
  =
  \varepsilon_n\log Z_n(\alpha).
\end{equation}
The KL dual problem can therefore be written as
\begin{equation}
\label{eq:kl-reduced-dual}
  \sup_{\substack{
  \alpha\in\mathbb R^n_{++}\\
  Z_n(\alpha)<\infty
  }}
  \left\{
  \frac1n\sum_{i=1}^n(1+\log\alpha_i)
  -\varepsilon_n\log Z_n(\alpha)
  \right\}.
\end{equation}

If there is no duality gap and the primal and dual optima are attained, then the exact primal solution has the exponential form
\begin{equation}
\label{eq:kl-gibbs-density}
  q_n^\star(\param)
  =
  \frac{
  \exp\left(
  \avLik{\alpha^\star}/\varepsilon_n
  \right)
  }{
  Z_n(\alpha^\star)
  },
  \qquad
  \priordist\text{-almost surely},
\end{equation}
where the dual variables satisfy the fixed-point equations
\begin{equation}
  \alpha_{i}^\star
  =
  \left\{
  \int_\Theta
  f_i(\param)
  \frac{
  \exp\left(
  \avLik{\alpha^\star}/\varepsilon_n
  \right)
  }{
  Z_n(\alpha^\star)
  }
  \,d\priordist(\param)
  \right\}^{-1}
  =
  \frac{1}{p_{Q_n^\star}(X_i)},
  \qquad
  i=1,\ldots,n.
\label{eq:kl-alpha-fixed-point}
\end{equation}
\end{corollary}

\begin{proof}
Substituting $\phi^*(s) = e^s - 1$ into \eqref{eq:finite-sample-dual} gives \eqref{eq:kl-dual-with-eta}.
For fixed $\alpha$, the $\eta$-dependent part of the dual objective is
\begin{multline*}
  -\eta
  -
  \varepsilon_n
  \int_\Theta
  \left[
  \exp\left(
  \frac{\avLik{\alpha}-\eta}{\varepsilon_n}
  \right)
  -1
  \right]
  \,d\priordist(\param) \\
  =-\eta-\varepsilon_n \exp\left(-\frac{\eta}{\varepsilon_n}\right)
  \int_\Theta\exp\left(\frac{\avLik{\alpha}}{\varepsilon_n}\right)d\priordist(\param)+\varepsilon_n.
\end{multline*}
Differentiating with respect to $\eta$ gives the first-order condition
\[
  -1
  +
  \exp\left(-\frac{\eta}{\varepsilon_n}\right)\int_\Theta
  \exp\left(
  \frac{\avLik{\alpha}}{\varepsilon_n}
  \right)
  d\priordist(\param)
  =
  0.
\]
Equivalently,
\[
  \exp\left(-\frac{\eta}{\varepsilon_n}\right)
  \int_\Theta
  \exp\left(
  \frac{\avLik{\alpha}}{\varepsilon_n}
  \right)
  d\priordist(\param)
  =
  1,
\]
which yields \eqref{eq:kl-optimal-eta}.
At this value of $\eta$,
\[
  \int_\Theta
  \left[
  \exp\left(
  \frac{\avLik{\alpha}-\eta(\alpha)}{\varepsilon_n}
  \right)
  -1
  \right]
  d\priordist(\param)
  =
  0,
\]
so \eqref{eq:kl-dual-with-eta} reduces to \eqref{eq:kl-reduced-dual}.
Finally, since $(\phi^*)'(s)=\exp(s)$,~\eqref{eq:kkt-q-differentiable} implies
\[
  q_n^\star(\param)
  =
  \exp\left(
  \frac{
  \avLik{\alpha^\star}-\eta^\star
  }{
  \varepsilon_n
  }
  \right),
\]
and hence \eqref{eq:kl-gibbs-density} is a consequence of \eqref{eq:kl-optimal-eta}.
The fixed-point equations \eqref{eq:kl-alpha-fixed-point} follow from \eqref{eq:kkt-alpha}.
\end{proof}

\subsection{Practical stability conditions}
\label{sec:practical_stability}

We next give verifiable sufficient conditions under which the value function \(\valfun\) is finite and continuous at the origin. By \Cref{prop:stable-strong-duality}, these conditions guarantee stability of the PrO problem; the subsequent remark describes common settings in which they are satisfied.

\begin{corollary}[Practical stability conditions]
\label[corollary]{cor:practical-stability}
Suppose that \Cref{ass:finite-sample-convex-setup} holds.
Assume that there exists $\bar q\in\dom\Phi$ such that
\[
  \int_\Theta \bar q\,d\priordist=1,
  \qquad
  \int_\Theta f_i\bar q\,d\priordist>0,
  \qquad
  i=1,\ldots,n,
\]
and such that
\[
  (1+r)\bar q\in\dom\Phi
\]
for every real number $r$ in a neighborhood of zero.
Suppose also that there exists $\alpha^0\in\mathbb R^n_{++}$ such that
\begin{equation}\label{eq:conj-fin}
\Phi^*\left(
\frac{\avLikPFree{\alpha^0}}{\varepsilon_n}
\right)
<
\infty.
\end{equation}
Then $\valfun$ is finite and continuous in a neighborhood of $(0,0)$.
Consequently, the PrO problem is stable: there is no duality gap and the supremum in the dual problem is attained.
If the infimum in the primal problem is also attained, the optimality conditions in Corollary~\ref{cor:kkt-system} hold.
\end{corollary}

\begin{proof}
For $r$ sufficiently close to zero, define \(q_r=(1+r)\bar q\).
Then \(\int q_r\,d\priordist=1+r\).
Since every coordinate of \(T_n\bar q\) is strictly positive, we have for all \((u,r)\) sufficiently close to \((0,0)\) that
\[
T_nq_r+u
=
(1+r)T_n\bar q+u
\in\mathbb R^n_{++}.
\]
Next, since $q_r\in\dom\Phi$ for $r$ near zero, it follows that
\[
  \valfun(u,r)
  \le
  \varepsilon_n\Phi(q_r)+g(T_nq_r+u)
  <
  +\infty
\]
for $(u,r)$ near $(0,0)$.
It remains to rule out $\valfun(u,r)=-\infty$ locally.
Fix $\alpha^0\in\mathbb R^n_{++}$ satisfying \eqref{eq:conj-fin}.
By
\eqref{eq:negative-log-conjugacy} applied coordinatewise,
\[
  g(T_nq+u)
  \ge
  -\int \avLikPFree{\alpha^0}q\,d\priordist
  -\frac1n\sum_{i=1}^n\alpha_i^0u_i
  +
  \frac1n\sum_{i=1}^n(1+\log\alpha_i^0).
\]
Moreover, by the Fenchel--Young inequality,
\[
  \varepsilon_n\Phi(q)
  -\int \avLikPFree{\alpha^0}q\,d\priordist
  \ge
  -\varepsilon_n
  \Phi^*
  \left(
  \frac{\avLikPFree{\alpha^0}}{\varepsilon_n}
  \right)
  >
  -\infty.
\]
Adding the preceding two inequalities ensures $\valfun(u,r)>-\infty$ for $(u,r)$ near the origin.
We have shown that $\valfun$ is real-valued in a neighborhood of $(0,0)$.
Since $\valfun$ is convex on the finite-dimensional perturbation space, it is continuous in that neighborhood and subdifferentiable at the origin.
The conclusion then follows from \Cref{prop:stable-strong-duality}.
\end{proof}

\begin{remark}[Restrictiveness of the stability condition]
The condition in \Cref{cor:practical-stability} is mild in the finite-sample problems considered here.
If the reference predictive density is positive at every observed point, then one may take $\bar q\equiv1$:
\[
  \int f_i\,d\priordist
  =
  p_\priordist(X_i)>0,
  \qquad i=1,\ldots,n.
\]
The rescaling condition is then satisfied by the usual $\phi$-divergences, including KL.
For KL regularization, the conjugate-finiteness condition \eqref{eq:conj-fin} is equivalent to
\[
  Z_n(\alpha^0)
  =
  \int_\Theta
  \exp\left(
  \frac{\avLik{\alpha^0}}{\varepsilon_n}
  \right)
  d\priordist(\param)
  <
  \infty
\]
for at least one $\alpha^0\in\mathbb R^n_{++}$.
This condition is automatic for finite mixtures and for compact parameter spaces with bounded likelihood functions.
Unbounded likelihood functions require either an appropriate integrability condition or a restriction of the parameter space.
\end{remark}

\section{Predictive risk error}
\label{sec:predictive-convergence-rates}

Given observed data \(X_1,\ldots,X_n\), PrO inference is an optimization problem over mixing distributions on the parameter space.
Recall that \(\empPro\) denotes an approximate solution to this problem and \(P_{\empPro}\) the predictive distribution it induces.
In this section, we study the behavior of \(P_{\empPro}\) in the large-data limit \(n\to\infty\).
While \cite{McLatchieCheriefAbdellatifFrazierKnoblauch2025} establishes broad predictive guarantees for PrO posteriors and compares them with Bayes and Gibbs distributions, our analysis focuses on predictive risk.
Specifically, we consider:
\((i)\) the predictive risk of the PrO approximate solution, \(\postPredRisk\), defined as the expected log-loss of the density $p_{\empPro}$ under the data-generating distribution;
\((ii)\) the benchmark predictive risk \(\optPredRisk\) over admissible mixing distributions having finite divergence from the prior; and
\((iii)\) the excess predictive risk
\[
  \risk(P_Q)-\optPredRisk.
\]
Our goal is to determine how quickly \(\postPredRisk\) approaches \(\optPredRisk\).

We derive a bound that decomposes the excess predictive risk into four effects: finite-data fluctuations, approximation of the benchmark risk by finite-divergence mixing distributions, regularization, and numerical optimization.
Our analysis does not require the infimum defining \(\optPredRisk\) to be attained.

\subsection{Excess predictive risk}

For a predictive distribution \(P\) with density \(p\), define its predictive risk by
\[
  \mathcal R(P)
  \coloneqq
  -\datadist\log p
  =
  -\int_{\mathsf X} \log p(x) \,d\datadist(x).
\]
When the relevant quantities are finite, differences in predictive risk are equivalent to differences in forward KL divergence from the data-generating distribution, namely,
\begin{equation}
\label{eq:risk-kl-relation}
  \mathcal R(P)-\mathcal R(P')
  =
  \KL(\datadist\|P)-\KL(\datadist\|P').
\end{equation}

Recall the admissible mixture class $\mathcal Q_{\mathcal E}$ from \eqref{eq:admissible-mixing-class}.
To compare predictive distributions generated by mixing distributions at different levels of divergence from the prior, for $\rho\ge0$, define
\begin{equation*}
\label{eq:Q-rho-definition}
  \mathcal Q_\rho
  \coloneqq
  \left\{
  Q\in\mathcal Q_{\mathcal E}:
  D_\phi(Q\|\priordist)\le\rho
  \right\},
  \qquad
  \Qinf
  \coloneqq
  \bigcup_{\rho<\infty}\mathcal Q_\rho.
\end{equation*}
The parameter \(\rho\) is a divergence budget: increasing \(\rho\) enlarges the comparison class and may improve its predictive performance but can also increase the regularization term in the PrO problem.
We think of $\mathcal Q_\rho$ as comparison classes rather than alternative estimators.
Define the benchmark predictive risk over the finite-divergence mixing class as
\begin{equation*}
\label{eq:R-star}
  \optPredRisk
  \coloneqq
  \inf_{Q\in\Qinf}\mathcal R(P_Q),
\end{equation*}
and assume $\optPredRisk<\infty$.
Note that the infimum defining $\optPredRisk$ need not be attained by a mixing distribution in $\Qinf$.
To quantify how well mixing distributions with a controlled divergence from the prior can approximate the benchmark predictive risk, define the approximation profile
\begin{equation}
\label{eq:A-rho-definition}
  \accProf(\rho)
  \coloneqq
  \inf_{Q\in\mathcal Q_\rho}
  \left\{
  \mathcal R(P_Q)-\optPredRisk
  \right\}.
\end{equation}

Since the sets $\mathcal Q_\rho$ increase with $\rho$, the approximation profile $\accProf$ is nonincreasing.
Moreover, \(\accProf(\rho)\downarrow0\) as \(\rho\to\infty\).
This is because, for every $\eta>0$, there is a \(Q_\eta\) such that
\[
  \risk(P_{Q_\eta})-\optPredRisk<\eta,
\]
so taking $\rho\ge D_\phi(Q_\eta\|\priordist)$ implies \(\accProf(\rho)<\eta\).
If \(\optPredRisk = \mathcal{R}(P_{Q^\star})\) for some \(Q^*\in \Qinf\), then
\[
  \accProf(\rho)=0
  \qquad
  \text{for every}
  \qquad
  \rho\ge D_\phi(Q^\star\|\priordist).
\]
Thus, there is no approximation cost once the comparison class contains a mixing distribution that attains the benchmark predictive risk.
The approximation profile \(\accProf\) is needed for the more general case in which the benchmark is not attained in \(\Qinf\), but can be approached by mixing distributions with increasingly large divergence budgets.

To control the effect of replacing the integral that defines \(\mathcal R\) with its finite-data counterpart, define
\begin{equation*}
\label{eq:Delta-n-infty}
  \Delta_n^\infty
  \coloneqq
  \sup_{Q\in\Qinf}
  |(P_n -\datadist)\ell_Q|.
\end{equation*}
This quantity is the largest difference between the empirical and expected log-loss over the mixing distributions having finite divergence with respect to the prior.

\begin{theorem}[Excess predictive risk bound]
\label{prop:penalized-oracle}
Let $\empPro$ satisfy \eqref{eq:approximate-primal-optimality}.
Then
\begin{equation}
\label{eq:penalized-global-oracle}
  \postPredRisk-\optPredRisk
  \leq
  \postPredRisk-\optPredRisk
  +
  \varepsilon_nD_\phi(\empPro\|\priordist)
  \le
  2\Delta^\infty_n
  +
  \inf_{\rho\ge0}
  \left\{
  \accProf(\rho)+\varepsilon_n\rho
  \right\}
  +
  \xi_n.
\end{equation}
\end{theorem}

\begin{proof}
Fix $\rho <\infty$ and $Q\in\mathcal Q_\rho$.
Approximate optimality gives
\[
  P_n\ell_{\empPro}
  +
  \varepsilon_nD_\phi(\empPro\|\priordist)
  \le
  P_n\ell_Q
  +
  \varepsilon_nD_\phi(Q\|\priordist)
  +
  \xi_n.
\]
Adding $(\datadist-P_n)\ell_{\empPro}-\optPredRisk$ to both sides and using
$D_\phi(Q\|\priordist)\le\rho$ yields
\begin{multline*}
  \postPredRisk-\optPredRisk
  +
  \varepsilon_nD_\phi(\empPro\|\priordist)
  \\
  \le
  (\datadist-P_n)\ell_{\empPro}
  +(P_n-\datadist)\ell_Q
  +\mathcal R(P_Q)-\optPredRisk
  +\varepsilon_n\rho
  +\xi_n.
\end{multline*}
Since $\empPro, Q\in\Qinf$, both finite-data terms are bounded above by $\Delta_n^\infty$.
Therefore,
\begin{equation*}
  \postPredRisk-\optPredRisk
  +
  \varepsilon_nD_\phi(\empPro\|\priordist)
  \le
  2\Delta_n^\infty
  +\mathcal R(P_Q)-\optPredRisk
  +\varepsilon_n\rho
  +\xi_n.
\end{equation*}
Taking the infimum over $Q\in\mathcal{Q}_\rho$ gives
\begin{equation*}
  \postPredRisk-\optPredRisk
  +
  \varepsilon_nD_\phi(\empPro\|\priordist)
  \le
  2\Delta_n^\infty
  +\accProf(\rho)
  +\varepsilon_n\rho
  +\xi_n,
\end{equation*}
and taking the infimum over $\rho\ge 0$ gives
\eqref{eq:penalized-global-oracle}.
\end{proof}

The infimum term in \eqref{eq:penalized-global-oracle} balances approximation and regularization.
Increasing \(\rho\) enlarges the comparison class \(\mathcal Q_\rho\), so the approximation profile \(\accProf(\rho)\) cannot increase, while the upper bound \(\varepsilon_n\rho\) for the regularization becomes larger.
Thus, Theorem~\ref{prop:penalized-oracle} separates the excess predictive risk into four effects, grouped into three terms: the finite-data fluctuation term \(\Delta_n^\infty\), the approximation--regularization term, and the numerical optimization error \(\xi_n\).
The following corollary gives corresponding bounds on the rate of convergence and sufficient conditions under which the predictive risk converges to the benchmark predictive risk.

\begin{corollary}[Rates for excess predictive risk]
\label[corollary]{cor:penalized-rates}
Suppose \(\Delta_n^\infty=O_p(r_n)\) and \(\xi_n=O_p(s_n)\).
Then
\begin{equation*}
\label{eq:generic-rate-conclusion}
  \postPredRisk-\optPredRisk
  =
  O_p\!\left(
  r_n
  +
  \inf_{\rho\ge0}
  \{\accProf(\rho)+\varepsilon_n\rho\}
  +
  s_n
  \right).
\end{equation*}
In particular:
\begin{enumerate}
\item If the benchmark predictive risk \(\optPredRisk\) is attained by a mixing distribution $Q^\star\in\Qinf$, then the infimum term is at most $\varepsilon_nD_\phi(Q^\star\|\priordist)$.

\item If $\accProf(\rho)\le C\rho^{-\beta}$ for all sufficiently large $\rho$, where $\beta>0$, then
\[
  \inf_{\rho\ge0}
  \{\accProf(\rho)+\varepsilon_n\rho\}
  =
  O\!\left(\varepsilon_n^{\beta/(\beta+1)}\right).
\]

\item If $\accProf(\rho)\le Ce^{-a\rho}$ for all sufficiently large $\rho$, where $a>0$, then
\[
  \inf_{\rho\ge0}
  \{\accProf(\rho)+\varepsilon_n\rho\}
  =
  O\!\left(
  \varepsilon_n\log\frac1{\varepsilon_n}
  \right).
\]
\end{enumerate}
If $\Delta_n^\infty=o_p(1)$, $\xi_n=o_p(1)$, and $\varepsilon_n\to0$, then
\begin{equation}
\label{eq:predictive-risk-consistency}
  \postPredRisk\rightarrow\optPredRisk
  \qquad\text{in probability}.
\end{equation}
The same conclusions hold almost surely when the corresponding bounds on the finite-data and optimization terms hold almost surely.
\end{corollary}

\begin{proof}
The rate statement follows from \eqref{eq:penalized-global-oracle}.
The three specific bounds follow by comparing with $Q^\star$, by choosing $\rho$ of order $\varepsilon_n^{-1/(\beta+1)}$, and by choosing $\rho=a^{-1}\log(1/\varepsilon_n)$, respectively.
It remains to verify that the infimum term converges to zero whenever $\varepsilon_n\to0$.
Let $\eta>0$.
By the definition of $\optPredRisk$, there exists $Q_\eta\in\Qinf$ such that
\[
  \mathcal R(P_{Q_\eta})-\optPredRisk<\frac\eta2.
\]
Set $\rho_\eta=D_\phi(Q_\eta\|\priordist)<\infty$.  Since
$Q_\eta\in\mathcal Q_{\rho_\eta}$,
\[
  \accProf(\rho_\eta)<\frac\eta2.
\]
For all sufficiently large $n$, also
$\varepsilon_n\rho_\eta<\eta/2$.
Therefore
\[
  \inf_{\rho\ge0}
  \{\accProf(\rho)+\varepsilon_n\rho\}
  \le
  \accProf(\rho_\eta)+\varepsilon_n\rho_\eta
  <\eta.
\]
Since $\eta$ is arbitrary, the deterministic term converges to zero.  The
probability and almost-sure conclusions now follow from \eqref{eq:penalized-global-oracle}.
\end{proof}

When \eqref{eq:risk-kl-relation} holds, the convergence in
\eqref{eq:predictive-risk-consistency} is equivalently convergence of the
forward-KL value:
\[
  \KL(\datadist\|P_{\empPro})
  \rightarrow
  \inf_{Q\in\Qinf}\KL(\datadist\|P_Q).
\]
This statement concerns only the convergence of predictive distributions.
It does not imply the convergence of the mixing distributions $\empPro$.

\section{Examples}
\label{sec:examples}

The two examples below illustrate complementary regimes of the preceding
theory.
The first is a categorical example in which the benchmark predictive risk is not attained by any mixing distribution on the original parameter space.
We derive sharp approximation rates under KL and \(\chi^2\) regularization and show that they lead to different rate-optimal regularization schedules.
We then characterize the limiting behavior of the predictive and mixing distributions.
Finally, combining these results with control of the finite-data fluctuation term yields end-to-end excess-risk
bounds through \Cref{prop:penalized-oracle}.
The second example considers a misspecified Gaussian location-mixture model.
After replacing the continuous parameter space by a fixed grid, we specialize \Cref{prop:penalized-oracle} to obtain an end-to-end excess-risk bound for the discretized model.
We then use the reduced KL dual problem and its optimality conditions to compute numerical solutions and assess their accuracy.

\subsection{Categorical example}
\label{sec:theoretical-example}

Let $\mathsf X=\{1,2,3\}$, $\Theta=(-1,1)$, and let $\priordist$ be uniform on $(-1,1)$.
Since \(\mathsf X\) is finite, we identify each probability distribution on \(\mathsf X\) with its probability mass vector, with coordinates ordered as (1,2,3).
For $\param\in(-1,1)$, define
\begin{equation*}
\label{eq:categorical-components}
  p_\param
  \coloneqq
  \left(
  \frac{3-\param^2+\param}{8},
  \frac{3-\param^2-\param}{8},
  \frac{1+\param^2}{4}
  \right),
\end{equation*}
and take
\begin{equation*}
\label{eq:categorical-truth}
  \datadist
  \coloneqq
  \left(\frac16,\frac16,\frac23\right).
\end{equation*}
For a mixing distribution $Q$, write
\[
  m_1(Q)=\int\param\,dQ(\param),
  \qquad
  m_2(Q)=\int\param^2\,dQ(\param).
\]
Then
\begin{equation}
\label{eq:categorical-mixture-moments}
  P_Q
  =
  \left(
  \frac{3-m_2(Q)+m_1(Q)}8,
  \frac{3-m_2(Q)-m_1(Q)}8,
  \frac{1+m_2(Q)}4
  \right).
\end{equation}
Since $m_2(Q)<1$ for every mixing distribution on $(-1,1)$, $P_Q(3)<1/2$, whereas $\datadist(3)=2/3$.
Thus, the data-generating distribution is not represented by any mixing distribution.
Let
\begin{equation*}
\label{eq:categorical-Pdagger}
  P^\dagger
  \coloneqq
  \left(\frac14,\frac14,\frac12\right).
\end{equation*}
To describe the relationship between $P^\dagger$ and the model class, we equip the probability distributions on $\mathsf X$ with the total-variation distance,
\begin{equation*}
\label{eq:tv-definition}
  \|P-P'\|_{\mathrm{TV}}
  \coloneqq
  \sup_{A\in\mathcal X}|P(A)-P'(A)|
  =
  \frac12\sum_{x\in\mathsf X}|P(x)-P'(x)|.
\end{equation*}
Let \(\mathcal M_\infty=\{P_Q:Q\in\Qinf\}\) and denote its closure in the total variation distance by \(\overline{\mathcal M}_\infty\).
The comparison class \(\Qinf\), and therefore potentially the benchmark predictive risk \(\optPredRisk\), depends on the choice of divergence.
In the present example, however, the KL and \(\chi^2\) divergences yield the same benchmark predictive risk and the same limiting predictive distribution.
The following proposition makes this precise.

\begin{proposition}[\(\optPredRisk\) not attained by a mixing distribution]
\label[proposition]{prop:categorical-oracle}
For each of the two choices
\[
  D_\phi=D_{\mathrm{KL}}
  \qquad\text{and}\qquad
  D_\phi=D_{\chi^2},
\]
the distribution $P^\dagger$ belongs to
$\overline{\mathcal M}_\infty$ and
\begin{equation}
\label{eq:categorical-risk-chain}
  \mathcal R(\datadist)
  <
  \mathcal R(P^\dagger)
  =
  \optPredRisk
  <
  \inf_{\param\in(-1,1)}\mathcal R(P_\param),
\end{equation}
but \(P^\dagger\) is not attained by any mixing
distribution on \((-1,1)\).
After extending the parameter space to \([-1,1]\), the unique mixing distribution representing \(P^\dagger\) is
\begin{equation}
\label{eq:categorical-Qdagger}
  Q^\dagger
  =
  \frac12\delta_{-1}+\frac12\delta_1.
\end{equation}
This distribution is singular with respect to the continuous prior and therefore has infinite KL and $\chi^2$ divergence from $\priordist$.
\end{proposition}

\begin{proof}
For $0<\delta<1$, let $Q_\delta$ be the distribution with one half of its mass placed uniformly on $(-1,-1+\delta)$ and the other half placed uniformly on $(1-\delta,1)$.
The density of this distribution with respect to $\priordist$ is $\delta^{-1}$ on these two intervals and zero elsewhere, so it has finite KL and $\chi^2$ divergence.
Symmetry gives
\[
  m_1(Q_\delta)=0,
  \qquad
  m_2(Q_\delta)=1-\delta+\frac{\delta^2}{3}.
\]
Consequently,
\begin{equation}
\label{eq:categorical-Qdelta-predictive}
  P_{Q_\delta}
  =
  \left(
  \frac14+\frac\delta8-\frac{\delta^2}{24},
  \frac14+\frac\delta8-\frac{\delta^2}{24},
  \frac12-\frac\delta4+\frac{\delta^2}{12}
  \right)
  \longrightarrow P^\dagger
\end{equation}
in total variation as $\delta\to 0$.
Hence $P^\dagger\in\overline{\mathcal M}_\infty$.
To show that $P^\dagger$ achieves the benchmark predictive risk \(\optPredRisk\), observe that
\[
  \sum_{x=1}^3
  \frac{\datadist(x)}{P^\dagger(x)}P_Q(x)
  =
  \frac56+\frac{m_2(Q)}6
  <1
\]
for every mixing distribution \(Q\) on \((-1,1)\). Hence, using \(\log u\ge 1-1/u\),
\begin{align*}
  \KL(\datadist\|P_Q)-\KL(\datadist\|P^\dagger)
  =
  \sum_{x=1}^3
  \datadist(x)\log\frac{P^\dagger(x)}{P_Q(x)}
  \ge
  1-
  \sum_{x=1}^3
  \frac{\datadist(x)}{P^\dagger(x)}P_Q(x)
  >0.
\end{align*}
Since \(P_{Q_\delta}\to P^\dagger\) and the KL divergence is continuous at
\(P^\dagger\),
\[
  \inf_{Q\in\Qinf}\KL(\datadist\|P_Q)
  =
  \KL(\datadist\|P^\dagger)
  \quad
  \iff
  \quad
  \optPredRisk = \mathcal{R}(P^\dagger),
\]
but the infimum is not attained by any $Q\in\Qinf$.
Since $\datadist\neq P^\dagger$, the strict positivity of the KL divergence implies \( \mathcal R(\datadist) < \mathcal R(P^\dagger) \), which is the leftmost inequality in \eqref{eq:categorical-risk-chain}.
A direct calculation shows
\begin{equation*}
\label{eq:component-gap}
  \inf_{\param\in(-1,1)}
  \left\{
  \mathcal R(P_\param)-\mathcal R(P^\dagger)
  \right\}
  =
  \frac16\log\frac43
  >0,
\end{equation*}
where the infimum is approached as \(\param\to\pm1\).
This proves the rightmost inequality in \eqref{eq:categorical-risk-chain}.
Finally, $P_Q=P^\dagger$ would require $m_1(Q)=0$ and $m_2(Q)=1$, which is impossible
on $(-1,1)$.
On $[-1,1]$, the condition $m_2(Q)=1$ forces all of the mass onto the two endpoints, and $m_1(Q)=0$ forces equal masses, establishing \eqref{eq:categorical-Qdagger}.
\end{proof}

The two inequalities in \eqref{eq:categorical-risk-chain} illustrate distinct features of prediction under model misspecification.
First, the benchmark predictive risk is strictly smaller than the predictive risk of every individual model \(P_\param\).
Second, the data-generating distribution is not represented by the predictive mixture class.
The first phenomenon is termed ``non-trivial misspecification'' in \cite{McLatchieCheriefAbdellatifFrazierKnoblauch2025}, whereas ``convex recovery'' refers to the stronger property that the data-generating distribution belongs to the predictive mixture class.
Thus, this example exhibits non-trivial misspecification without convex recovery.

The predictive convergence guarantee for non-trivial misspecification in \cite{McLatchieCheriefAbdellatifFrazierKnoblauch2025} does not require convex recovery, but it assumes the existence of a mixing distribution having finite KL divergence from the prior and whose induced predictive distribution attains the benchmark predictive risk \(\optPredRisk\).
The present example falls outside this setting.
Extending the parameter space to \([-1,1]\) allows \(P^\dagger\) to be represented by the \(Q^\dagger\) in
\eqref{eq:categorical-Qdagger}.
Since \(Q^\dagger\) is singular with respect to the continuous prior, it has infinite KL and \(\chi^2\) divergence.
Thus, the extension provides a representation of the predictive limit without changing \(\mathcal Q_\infty\).

\subsubsection{Divergence-specific approximation rates}
\label{subsec:divergence-accessibility}

For the sequence of mixing distributions \(Q_\delta\) constructed above,
\begin{equation*}
\label{eq:Qdelta-divergences}
  D_{\mathrm{KL}}(Q_\delta\|\priordist)
  =
  \log(1/\delta)
  \quad
  \text{and}
  \quad
  D_{\chi^2}(Q_\delta\|\priordist)
  \coloneqq \int
  \left(
  \frac{dQ_\delta}{d\priordist}-1
  \right)^2
  d\priordist =
  \delta^{-1}-1.
\end{equation*}
Using \eqref{eq:categorical-Qdelta-predictive}, the corresponding excess predictive risk is
\begin{equation*}
  \mathcal R(P_{Q_\delta})-\mathcal R(P^\dagger)
  =
  -\frac13
  \log\left(
  1+\frac{\delta}{2}-\frac{\delta^2}{6}
  \right)
  -\frac23
  \log\left(
  1-\frac{\delta}{2}+\frac{\delta^2}{6}
  \right)
  =
  \frac{\delta}{6}+O(\delta^2),
\label{eq:Qdelta-risk-expansion}
\end{equation*}
where the final equality follows from a Taylor expansion at \(\delta=0\).
Choosing \(\delta=e^{-\rho}\) for KL divergence and \(\delta=(1+\rho)^{-1}\) for \(\chi^2\) divergence gives, for all sufficiently large \(\rho\),
\begin{equation}
\label{eq:elementary-accessibility-bounds}
  \accProfKL(\rho)
  \le
  Ce^{-\rho}
  \quad
  \text{and}
  \quad
  \accProfChi(\rho)
  \le
  \frac{C}{1+\rho},
\end{equation}
where \(\accProfKL\) and \(\accProfChi\) denote the approximation profiles defined using KL and \(\chi^2\) divergence, respectively.
The next result shows that these upper bounds are sharp and, in the KL case, identifies the leading constant.

\begin{proposition}[Sharp approximation rates]
\label[proposition]{prop:sharp-accessibility}
For the KL divergence,
\begin{equation}
\label{eq:sharp-KL-accessibility}
  \lim_{\rho\to\infty}
  e^\rho \accProfKL(\rho)
  =
  \frac{1}{3e}.
\end{equation}
For the \(\chi^2\) divergence, there exist constants \(0<c<C<\infty\) and \(\rho_0>0\) such that
\begin{equation}
\label{eq:sharp-pearson-accessibility}
  \frac{c}{\rho}
  \le
  \accProfChi(\rho)
  \le
  \frac{C}{\rho}
  \quad
  \text{for every}
  \quad
  \rho\ge\rho_0.
\end{equation}
\end{proposition}

The proof is deferred to \Cref{app:sharp-accessibility} in the appendix.
The divergence-specific approximation rates in \Cref{prop:sharp-accessibility} lead to different approximation--regularization trade-offs.
By
\Cref{cor:penalized-rates} and \Cref{prop:sharp-accessibility},
\begin{equation*}
\label{eq:KL-Pearson-tradeoffs}
  \inf_{\rho\ge0}
  \{\accProfKL(\rho)+\varepsilon\rho\}
  =
  O\!\left(\varepsilon\log\frac{1}{\varepsilon}\right)
  \quad\text{and}\quad
  \inf_{\rho\ge0}
  \{\accProfChi(\rho)+\varepsilon\rho\}
  =
  O(\varepsilon^{1/2}).
\end{equation*}
Balancing a finite-data contribution of order \(n^{-1/2}\) with the deterministic approximation--regularization term gives
\begin{equation}
\label{eq:divergence-specific-epsilon}
  \varepsilon_n
  =
  \frac{1}{\sqrt n\log n}
  \quad\text{for KL}
  \quad
  \text{and}
  \quad
  \varepsilon_n
  =
  \frac1n
  \quad\text{for \(\chi^2\)}.
\end{equation}
The two regularization schedules reflect the different costs of concentrating mass on sets having small prior probability: the KL divergence grows logarithmically in the reciprocal prior mass, whereas the \(\chi^2\) divergence grows at a reciprocal-mass rate.
More generally, the approximation rate is determined jointly by the prior mass of parameter regions that support mixing distributions with near-benchmark predictive risk and by the cost, under the chosen divergence, of concentrating on those regions.

\subsubsection{Predictive and mixing-distribution limits}
\label{subsec:catpredconvergence}

We next show that for the categorical example, convergence of the predictive risk has different implications for the predictive and mixing distributions.
The predictive distributions converge to \(P^\dagger\), whereas the mixing distributions move their mass toward the boundary of the parameter space and have no admissible weak limit.

\begin{proposition}[Predictive convergence and boundary concentration]
\label[proposition]{prop:categorical-boundary-convergence}
Let \(\{Q_n\}\) be any sequence of mixing distributions on $(-1,1)$ such that
\(
  \mathcal R(P_{Q_n})\rightarrow\mathcal R(P^\dagger).
\)
Then
\begin{equation}
\label{eq:categorical-predictive-TV}
  \|P_{Q_n}-P^\dagger\|_{\mathrm{TV}}
  \rightarrow0,
\end{equation}
whereas for every $\eta\in(0,1)$,
\begin{equation}
\label{eq:boundary-concentration}
  Q_n\{\,|\param|\le1-\eta\,\}
  \rightarrow0.
\end{equation}
Thus, \(\{Q_n\}\) has no weak limit that is a probability distribution on \((-1,1)\).
If instead each \(Q_n\) is regarded as a probability distribution on \([-1,1]\), then
\begin{equation}
\label{eq:compactified-weak-limit}
  Q_n
  \rightsquigarrow
  \frac12\delta_{-1}+\frac12\delta_1.
\end{equation}
\end{proposition}

\begin{proof}
By \eqref{eq:categorical-mixture-moments}, \(P_Q\) is determined by the tuple \((m_1(Q),m_2(Q))\).
Every subsequence of \(\{(m_1(Q_n),m_2(Q_n))\}\)
has a further subsequence converging in \([-1,1]\times[0,1]\).
Continuity of the predictive risk and uniqueness of
the minimizing predictive distribution \(P^\dagger\) force $m_1(Q_n)\to0$ and $m_2(Q_n)\to1$, which proves
\eqref{eq:categorical-predictive-TV}.
For fixed $\eta>0$,
\[
  1-m_2(Q_n)
  =
  \int(1-\param^2)\,dQ_n(\param)
  \ge
  \left(1-(1-\eta)^2\right)
  Q_n\{|\param|\le1-\eta\},
\]
which proves \eqref{eq:boundary-concentration}.
A weak probability limit on $(-1,1)$ would assign zero mass to every set $\{|\param|<1-\eta\}$ and hence zero mass to their union, which is all of $(-1,1)$; therefore no such limit exists.
On $[-1,1]$, every subsequential weak limit is supported on the two endpoints.
The condition $m_1(Q_n)\to0$ forces equal endpoint masses, so the only possible limit is \eqref{eq:compactified-weak-limit}.
\end{proof}

To specialize \Cref{prop:penalized-oracle} to the categorical example, it remains to control the finite-data fluctuation term $\Delta_n^\infty$.
Every coordinate of \(p_\param\) is at least \(1/8\), and this lower bound is preserved under mixing.
Hence \(0\le\ell_Q(x)\le\log 8\) for every mixing distribution \(Q\) and every \(x\in\mathsf X\), which gives
\begin{equation}
\label{eq:categorical-Delta-bound}
  \Delta_n^\infty
  \le
  (\log 8)
  \sum_{x=1}^3
  |P_n(x)-\datadist(x)|.
\end{equation}
The multinomial central limit theorem therefore implies \(\Delta_n^\infty=O_p(n^{-1/2})\).
The strong law of large numbers and \eqref{eq:categorical-Delta-bound} also ensure \(\Delta_n^\infty\to0\) almost surely.

If a computed solution has a primal optimal value error of \(g_n\), then it satisfies \eqref{eq:approximate-primal-optimality} with \(\xi_n=g_n\).
Combining the $\Delta_n^\infty$ bound with \Cref{prop:penalized-oracle,prop:sharp-accessibility} yields
\begin{equation*}
  \postPredRisk-\optPredRisk
  =
  O_p\!\left(
  n^{-1/2}
  +
  \varepsilon_n\log(1/\varepsilon_n)
  +
  g_n
  \right)
\label{eq:categorical-end-to-end-KL}
\end{equation*}
under KL regularization and
\begin{equation*}
  \postPredRisk-\optPredRisk
  =
  O_p\!\left(
  n^{-1/2}
  +
  \varepsilon_n^{1/2}
  +
  g_n
  \right)
\label{eq:categorical-end-to-end-Pearson}
\end{equation*}
under \(\chi^2\) regularization.
With the divergence-specific regularization schedules in \eqref{eq:divergence-specific-epsilon} and \(g_n=O_p(n^{-1/2})\), both excess-risk bounds are of order \(n^{-1/2}\).

\subsection{Gaussian mixture example}
\label{sec:numerics-dual-misspecification}

\newcommand{\DualNumericsK}{100}
\newcommand{\DualNumericsThetaMin}{-5}
\newcommand{\DualNumericsThetaMax}{5}
\newcommand{\DualNumericsSigmaModel}{0.3}
\newcommand{\DualNumericsReps}{100}
\newcommand{\DualNumericsRepresentativeN}{4096}
\newcommand{\DualNumericsMaxN}{4096}
\newcommand{\DualNumericsNs}{8, 16, 32, 64, 128, 256, 512, 1024, 2048, 4096}
\newcommand{\DualNumericsProjectionKL}{0.0201}
\newcommand{\DualNumericsMedianFinalExcessRisk}{\ensuremath{2.54\times 10^{-3}}}
\newcommand{\DualNumericsMaxKKTResidual}{\ensuremath{9.71\times 10^{-5}}}
\newcommand{\DualNumericsMedianKKTResidual}{\ensuremath{2.70\times 10^{-6}}}
\newcommand{\DualNumericsMaxDualityGap}{\ensuremath{4.03\times 10^{-11}}}
\newcommand{\DualNumericsMedianDualityGap}{\ensuremath{2.60\times 10^{-13}}}
\newcommand{\DualNumericsRepresentativeKKTResidual}{\ensuremath{9.22\times 10^{-6}}}
\newcommand{\DualNumericsRepresentativeDualityGap}{\ensuremath{4.90\times 10^{-13}}}
\newcommand{\DualNumericsCertifiedRate}{1.000}

The categorical example treats a setting in which the benchmark predictive risk is not attained.
We now consider a misspecified Gaussian location-mixture model and replace its continuous parameter space by a fixed grid for computation.
In the resulting finite-cardinality model, the benchmark predictive risk is attained.
This second example connects the predictive-convergence and duality results.
We first specialize the excess-risk bound in \Cref{prop:penalized-oracle} to the discretized model, obtaining a theoretical benchmark for the convergence of the predictive risk.
Then we solve the reduced KL dual problem from \Cref{cor:kl-gibbs-form}, recover the fitted mixture weights from the dual variables, and assess numerical accuracy using the optimality conditions in \Cref{cor:kl-gibbs-form} and the primal and dual optimal values.
Finally, we compare the observed decay of the excess predictive risk with the rate supplied by the specialized bound.

The example is not intended to compare primal and dual algorithms.
Rather, it illustrates the dimension reduction provided by the dual formulation: the primal problem is posed over mixing distributions, whereas the reduced dual problem has one coordinate per observation.

The data-generating distribution is the skewed, heavy-tailed mixture
\[
  \datadist
  =
  0.55\,t_5(-1.2,0.75)
  +
  0.30\,N(1.1,0.35^2)
  +
  0.15\,N(2.8,0.90^2).
\]
Here, \(t_5(\mu,s)\) is a Student-\(t\) distribution with five degrees of freedom, location \(\mu\), and scale \(s\).
We write \(N(\mu,\sigma^2)\) for a normal distribution with mean \(\mu\) and variance \(\sigma^2\).
Set
\[
  \Theta=[\DualNumericsThetaMin,\DualNumericsThetaMax]
  \quad
  \text{and}
  \quad
  P_\param = N(\param,\DualNumericsSigmaModel^2).
\]
The induced predictive distributions are Gaussian location mixtures with variances \(\DualNumericsSigmaModel^2\) and locations restricted to \(\Theta\).
This mixture class cannot reproduce the polynomial tails of the Student-\(t\) component of $\datadist$.
Thus, there is no mixing distribution $Q$ on $\Theta$ for which $P_Q=\datadist$.

For computation, we replace the continuum \(\Theta\) by
\(K=\DualNumericsK\) equally spaced locations \(\param_1,\ldots,\param_K\) and take the uniform grid prior
\(\priordist_j=1/K\).
Let
\[
  W_K
  \coloneqq
  \left\{
  w\in\mathbb R_+^K:
  \sum_{j=1}^K w_j=1
  \right\}.
\]
A mixing vector \(w\in W_K\) induces the predictive distribution \(P_w\) with density
\[
  p_w(x)
  =
  \sum_{j=1}^K
    w_j\frac{1}{\sigma}\varphi\left(\frac{x-\param_j}{\sigma}\right),
    \quad \sigma := \DualNumericsSigmaModel,
\]
where $\varphi$ is the standard normal density.
Let \(P_{w^\dagger}\) be a grid-restricted KL projection of \(\datadist\), namely, let
\[
  w^\dagger
  \in
  \argmin_{w\in W_K}
  \left\{- \datadist \log p_w \right\}.
\]
The \(\datadist\)-expectations in population quantities like this are evaluated by quadrature.
For the grid-restricted KL projection,
\[
  \KL(\datadist\|P_{w^\dagger})
  \approx
  \DualNumericsProjectionKL.
\]

\subsubsection{Specialization of Theorem~\ref{prop:penalized-oracle}}
For the discretized problem,
\[
  \Delta_n^\infty
  =
  \sup_{w\in W_K}
  \left|
  (P_n-\datadist)(-\log p_w)
  \right|
  =
  \sup_{w\in W_K}
  \left|
  (P_n-\datadist)\log p_w
  \right|.
\]
To control this term, we first obtain a common bound for the losses.
Since
\(|\param_j|\le 5\), each Gaussian component satisfies
\[
  \frac{1}{\sigma\sqrt{2\pi}}
  \exp\left\{
  -\frac{(|x|+5)^2}{2\sigma^2}
  \right\}
  \le
  p_{\param_j}(x)
  =
  \frac{1}{\sigma}
  \varphi\left(
  \frac{x-\param_j}{\sigma}
  \right)
  \le
  \frac{1}{\sigma\sqrt{2\pi}}.
\]
Since \(p_w\) is a convex combination of the component densities, the
same bounds hold for \(p_w\).
It follows that there exists \(C<\infty\) such that
\begin{equation}
  \label{eq:numerical-loss-envelope}
  \sup_{w\in W_K}|\log p_w(x)|
  \le
  C(1+x^2).
\end{equation}
The data-generating distribution \(\datadist\) has a finite fourth moment, so the right-hand side of \eqref{eq:numerical-loss-envelope} is square-integrable under \(\datadist\).
In addition to this square-integrability, we must control the complexity of the losses.
Every mixture density is a linear combination of the \(K\) fixed
component densities:
\[
  p_w
  =
  \sum_{j=1}^K w_jp_{\param_j}.
\]
The density class is therefore contained in a vector space of dimension at most \(K\) and is VC-subgraph by \cite[Lemma~2.6.16]{vanDerVaartWellner2023}.
Since \(s\mapsto-\log s\) is monotone, the loss class
\[
  \{-\log p_w:w\in W_K\}
\]
is also VC-subgraph by \cite[Lemma~2.6.20 ($viii$)]{vanDerVaartWellner2023}.
The covering bound for VC-subgraph classes
\cite[Theorem~2.6.7]{vanDerVaartWellner2023}, together with the
square-integrable bound above, verifies the conditions of
\cite[Theorem~2.5.2]{vanDerVaartWellner2023}. Consequently, for fixed \(K\),
\[
  \sqrt{n}\Delta_n^\infty = \sup_{w\in W_K}
  \left|
  \sqrt{n}(P_n-\datadist)\log p_w
  \right|
  =
  O_p(1).
\]
The population minimizer \(w^\dagger\) belongs to \(W_K\), and
the uniform grid prior satisfies
\[
  \KL(w^\dagger\|\priordist)
  \le
  \log K,
\]
so applying \Cref{prop:penalized-oracle} with \(w^\dagger\) gives
\begin{equation}
\label{eq:numerical-excess-risk-rate}
  \mathcal R(P_{\widehat w_n^{\mathrm{dual}}})
  -
  \mathcal R(P_{w^\dagger})
  =
  O_p\left(
  n^{-1/2}+\varepsilon_n+\xi_n
  \right),
\end{equation}
which motivates setting $\varepsilon_n = n^{-1/2}$.

\subsubsection{Numerical results}

For samples \(X_1,\ldots,X_n\overset{\text{iid}}{\sim} \datadist\), we compute the specialization of the reduced KL dual problem in \eqref{eq:kl-reduced-dual}.
We have
\[
  \widehat\alpha_n
  \in
  \argmax_{\alpha\in\mathbb R^n_{++}}
  \left[
  \frac1n\sum_{i=1}^n(1+\log\alpha_i)
  -
  \varepsilon_n
  \log\left\{
  \sum_{j=1}^K
  \Pi_j
  \exp\left(\frac{(\alpha,f(\param_j))_n}{\varepsilon_n}\right)
  \right\}
  \right],
  \quad
\varepsilon_n=n^{-1/2}.
\]
The corresponding mixture weights follow from \eqref{eq:kl-gibbs-density}:
\[
  \widehat w_{n,j}^{\mathrm{dual}}
  =
  \frac{
  \Pi_j\exp\{(\widehat\alpha_n, f(\param_j))_n/\varepsilon_n\}
  }{
  \sum_{\ell=1}^K
  \Pi_\ell\exp\{(\widehat\alpha_n, f(\param_\ell))_n/\varepsilon_n\}
  },
  \qquad j=1,\ldots,K.
\]
The reduced dual problem is solved numerically using L-BFGS-B \cite{ByrdLuNocedalZhu1995,MoralesNocedal2011,ZhuByrdLuNocedal1997}.

We use the sample sizes \(n\in\{\DualNumericsNs\}\) and repeat the experiment over \(\DualNumericsReps\) independent data sets.
For every fitted predictive distribution, the reported excess risk is
\[
  \mathcal R(P_{\widehat w_n^{\mathrm{dual}}})
  -
  \mathcal R(P_{w^\dagger}).
\]

\begin{figure}[!htbp]
\centering
\includegraphics[width=\textwidth]{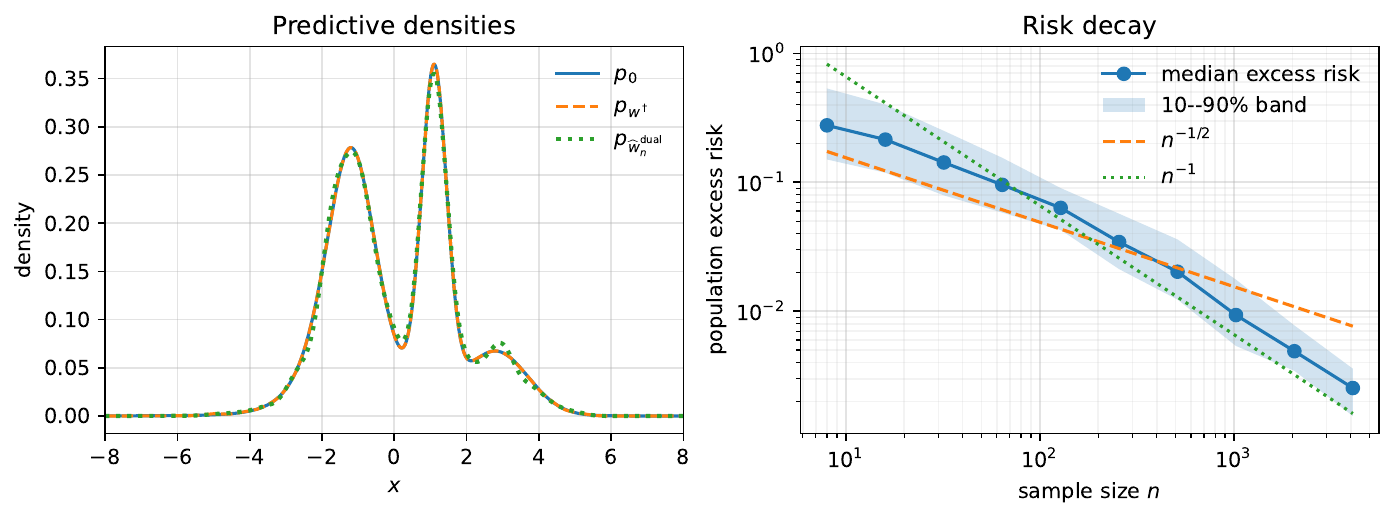}
\caption{Dual computation under misspecification.
Left: the true density \(\datadens\), its KL projection \(p_{w^\dagger}\) onto the \(K=\DualNumericsK\) component grid mixture class, and the fitted predictive density \(p_{\widehat w_n^{\mathrm{dual}}}\) obtained from one sample with \(n=\DualNumericsRepresentativeN\).
Right: excess risk over \(\DualNumericsReps\) independently generated data sets.
The solid curve is the pointwise median, and the shaded region
is the 10--90\% band; the reference lines are proportional to \(n^{-1/2}\) and \(n^{-1}\).
}
\label{fig:dual-k100-illustration}
\end{figure}

The left panel of \Cref{fig:dual-k100-illustration} illustrates two sources of error.
The difference between \(\datadens\) and \(p_{w^\dagger}\)
reflects approximation error by the chosen grid mixture class, whereas the
difference between \(p_{\widehat w_n^{\mathrm{dual}}}\) and
\(p_{w^\dagger}\) reflects finite-sample estimation error.
The fitted predictive density \(p_{\widehat w_n^{\mathrm{dual}}}\) is close to the projected predictive density \(p_{w^\dagger}\) over the central part of the distribution.
The right panel compares the decay of the population excess risk with two reference rates.
The \(n^{-1/2}\) line is the conservative rate supplied by \eqref{eq:numerical-excess-risk-rate} for \(\varepsilon_n=n^{-1/2}\) and when the numerical optimization error is negligible.
The \(n^{-1}\) line is included as a visual benchmark; it has not been established theoretically.
Across the displayed runs, the largest absolute deviation from the \(\alpha\) fixed-point condition in \eqref{eq:kl-alpha-fixed-point} is \(\DualNumericsMaxKKTResidual\), and the largest absolute difference between the primal and dual optimal values is \(\DualNumericsMaxDualityGap\).
These small optimization residuals indicate that numerical error is unlikely to materially affect the patterns shown in the figure.

\section{Conclusion}\label{sec:discussion-conclusion}

This paper analyzes predictively oriented (PrO) inference under the logarithmic score and a broad class of $\phi$-divergence penalties.
This specializes the general scoring-rule framework of \cite{McLatchieCheriefAbdellatifFrazierKnoblauch2025} while broadening the regularization beyond the KL divergence.
The PrO problem is posed over probability distributions on the parameter space and may therefore be infinite dimensional.
The dual problem, by contrast, has only $n+1$ scalar variables: one variable for each observation and one normalization multiplier.

Under stability conditions derived in \Cref{sec:finite-sample-duality}, the primal and dual optimal values agree and the dual optimum is attained.
When a primal solution also exists, the optimality conditions recover its density relative to the prior from the dual variables and provide fixed-point and primal--dual certificates for assessing the accuracy of numerical solutions.
For KL regularization, the normalization multiplier can be eliminated analytically, and the recovered density has an explicit exponential form.

The analysis in \Cref{sec:predictive-convergence-rates} accommodates approximate PrO posteriors and does not require the benchmark predictive risk to be attained.
The resulting excess-risk bound isolates the effects of sampling, approximation under a finite divergence budget, regularization, and numerical error.
This decomposition yields conditions for predictive-risk consistency and provides convergence rates when the individual terms can be controlled.

Taken together, the two examples in \Cref{sec:examples} show that the theory applies in regimes with and without attainment of the benchmark predictive risk.
The categorical example demonstrates that predictive distributions may converge to a unique limit even when the underlying mixing distributions have no weak limit on the original parameter space.
It also shows that the choice of divergence can materially affect the approximation--regularization tradeoff and the regularization schedule.
The misspecified Gaussian location-mixture example demonstrates how the
predictive-risk analysis and the dual PrO problem can be used together.

Several directions remain open within the framework developed here.
First, obtaining sharper rates from the excess-risk bound requires model-specific control of the empirical-process term and the approximation profile.
Second, although conditions relating predictive convergence to convergence of the mixing distributions have been studied for PrO posteriors \cite{McLatchieCheriefAbdellatifFrazierKnoblauch2025}, it would be useful to understand how such conditions interact with nonattainment and with sequences whose divergences grow without bound.
Third, since the dual formulation replaces an infinite-dimensional primal problem with an optimization problem whose dimension grows linearly with the sample size, developing algorithms that exploit the structure of this formulation at scale is an important computational direction.
Finally, extending duality and predictive-risk analyses under nonattainment from the logarithmic score to other proper scoring rules is a natural direction for future work.

\begin{appendix}

\section{Proof of Proposition~\ref{prop:sharp-accessibility}}
\label{app:sharp-accessibility}

\begin{proof}
The mixing distributions \(Q_\delta\) establish the upper bounds in
\eqref{eq:elementary-accessibility-bounds} for both divergences.

We first sharpen the KL bound by solving the corresponding constrained
problem. For a mixing distribution \(Q\) on \((-1,1)\), let \(Q^{-}\) denote
its image under \(\param\mapsto-\param\), and set
\[
\overline Q
=
\frac12(Q+Q^{-}).
\]
Since the prior is symmetric,
\(
D_{\mathrm{KL}}(Q^{-}\|\priordist)
=
D_{\mathrm{KL}}(Q\|\priordist)
\),
and convexity of KL divergence in its first argument gives
\(
D_{\mathrm{KL}}(\overline Q\|\priordist)
\le
D_{\mathrm{KL}}(Q\|\priordist)
\).
Reflection interchanges the first two coordinates of \(P_Q\). Since the first
two coordinates of \(\datadist\) are equal,
\(
\mathcal R(P_{Q^{-}})=\mathcal R(P_Q)
\).
Moreover,
\(
P_{\overline Q}
=
\frac12(P_Q+P_{Q^{-}})
\),
so convexity of predictive risk gives
\[
\mathcal R(P_{\overline Q})
\le
\frac12\mathcal R(P_Q)
+
\frac12\mathcal R(P_{Q^{-}})
=
\mathcal R(P_Q).
\]
It is therefore enough to consider symmetric mixing distributions.
For \(m\in[0,1]\), define
\[
r(m)
\coloneqq
-\frac13\log\frac{3-m}{8}
-\frac23\log\frac{1+m}{4}.
\]
If \(Q\) is symmetric, then \(m_1(Q)=0\) and
\(
\mathcal R(P_Q)=r(m_2(Q))
\).
Furthermore,
\[
r'(m)
=
\frac{3m-5}{3(3-m)(1+m)}
<0,
\qquad
0\le m\le1.
\]
Thus, among symmetric mixing distributions, minimizing predictive risk is
equivalent to maximizing the second moment.
For \(\lambda>0\), define
\begin{equation*}
\label{eq:Qlambda-definition}
Z_\lambda
\coloneqq
\int_{-1}^1
\exp(\lambda u^2)\,d\priordist(u)
\quad
\text{and}
\quad
\frac{dQ_\lambda}{d\priordist}(\param)
=
\frac{\exp(\lambda\param^2)}{Z_\lambda},
\end{equation*}
and write, within this proof,
\[
m_\lambda
\coloneqq
\int\param^2\,dQ_\lambda(\param)
\quad
\text{and}
\quad
\rho_\lambda
\coloneqq
D_{\mathrm{KL}}(Q_\lambda\|\priordist).
\]
For any \(Q\ll\priordist\),
\begin{equation}
\label{eq:dkl_inequal_cat}
0
\le
D_{\mathrm{KL}}(Q\|Q_\lambda)
=
D_{\mathrm{KL}}(Q\|\priordist)
-\lambda\int\param^2\,dQ(\param)
+\log Z_\lambda.
\end{equation}
For \(Q=Q_\lambda\),
\(
\rho_\lambda
=
\lambda m_\lambda-\log Z_\lambda
\).
Combining this identity with \eqref{eq:dkl_inequal_cat} gives
\begin{equation*}
\label{eq:moment-KL-comparison}
\lambda
\left\{
\int\param^2\,dQ(\param)-m_\lambda
\right\}
\le
D_{\mathrm{KL}}(Q\|\priordist)-\rho_\lambda.
\end{equation*}
Hence, every mixing distribution \(Q\) satisfying
\(D_{\mathrm{KL}}(Q\|\priordist)\le\rho_\lambda\) has
\(m_2(Q)\le m_\lambda\).
To see that \(Q_\lambda\) solves the predictive-risk problem at radius
\(\rho_\lambda\), let \(Q\) be any such mixing distribution. Its
symmetrization \(\overline Q\) is also feasible, satisfies
\(m_2(\overline Q)=m_2(Q)\), and has no larger predictive risk. Since \(r\) is
strictly decreasing,
\[
\mathcal R(P_Q)
\ge
\mathcal R(P_{\overline Q})
=
r(m_2(Q))
\ge
r(m_\lambda)
=
\mathcal R(P_{Q_\lambda}).
\]
Therefore, \(Q_\lambda\) minimizes predictive risk over all mixing
distributions with KL divergence at most \(\rho_\lambda\).
Since
\(
m_\lambda=(\log Z_\lambda)'
\),
differentiation gives
\[
\frac{d\rho_\lambda}{d\lambda}
=
\lambda\operatorname{Var}_{Q_\lambda}(\param^2)
>0.
\]
Thus, \(\lambda\mapsto\rho_\lambda\) is continuous and strictly increasing.
Since \(\priordist\) is uniform on \((-1,1)\),
\[
Z_\lambda
=
\int_0^1 e^{\lambda u^2}\,du.
\]
An expansion at the endpoint \(u=1\) gives
\[
Z_\lambda
=
\frac{e^\lambda}{2\lambda}
\left\{1+O(\lambda^{-1})\right\},
\qquad
m_\lambda
=
1-\lambda^{-1}+O(\lambda^{-2}),
\]
and therefore
\begin{equation}
\label{eq:rho-lambda-asymptotic}
\rho_\lambda
=
\log(2\lambda)-1+O(\lambda^{-1}).
\end{equation}
In particular, \(\rho_\lambda\downarrow0\) as \(\lambda\downarrow0\) and
\(\rho_\lambda\uparrow\infty\) as \(\lambda\uparrow\infty\), so every positive
radius corresponds to a unique value of \(\lambda\).
Since \(r(1)=\mathcal R(P^\dagger)\) and \(r'(1)=-1/6\), a Taylor expansion of
\(r\) at \(m=1\) gives
\begin{equation}
\label{eq:risk-lambda-asymptotic}
\mathcal R(P_{Q_\lambda})-\mathcal R(P^\dagger)
=
r(m_\lambda)-r(1)
=
\frac{1}{6\lambda}+O(\lambda^{-2}).
\end{equation}
Since \(Q_\lambda\) solves the constrained problem at radius
\(\rho_\lambda\),
\[
\accProfKL(\rho_\lambda)
=
\mathcal R(P_{Q_\lambda})-\mathcal R(P^\dagger).
\]
Combining \eqref{eq:rho-lambda-asymptotic} and
\eqref{eq:risk-lambda-asymptotic} yields
\[
e^{\rho_\lambda}\accProfKL(\rho_\lambda)
=
\frac{1}{3e}+o(1),
\]
which proves \eqref{eq:sharp-KL-accessibility}.

We next prove the lower bound for \(\chi^2\) divergence. Let
\(Q\ll\priordist\) satisfy
\(D_{\chi^2}(Q\|\priordist)\le\rho\). Then
\[
\int
\left(
\frac{dQ}{d\priordist}
\right)^2
d\priordist
\le
1+\rho.
\]
For \(0<\delta<1\), let
\(
B_\delta=\{\param:|\param|>1-\delta\}
\).
Since \(\priordist(B_\delta)=\delta\), the Cauchy--Schwarz inequality gives
\[
Q(B_\delta)
\le
\sqrt{(1+\rho)\delta}.
\]
Choose
\(
\delta=\{4(1+\rho)\}^{-1}
\).
Then \(Q(B_\delta)\le1/2\), so at least one half of the mass of \(Q\) lies
where \(|\param|\le1-\delta\). Consequently,
\[
1-m_2(Q)
=
\int(1-\param^2)\,dQ(\param)
\ge
\frac{1}{8(1+\rho)}.
\]
The preceding symmetrization argument also applies to the \(\chi^2\) divergence: reflection preserves the divergence, and convexity in its first argument implies
\(
D_{\chi^2}(\overline Q\|\priordist)
\le
D_{\chi^2}(Q\|\priordist)
\).
Moreover, \(m_2(\overline Q)=m_2(Q)\) and
\(\mathcal R(P_{\overline Q})\le\mathcal R(P_Q)\). Since
\(
-r'(m)\ge1/6
\)
for \(0\le m\le1\), the mean value theorem gives
\[
\mathcal R(P_Q)-\mathcal R(P^\dagger)
\ge
r(m_2(Q))-r(1)
\ge
\frac16\{1-m_2(Q)\}
\ge
\frac{1}{48(1+\rho)}.
\]
Together with the upper bound in
\eqref{eq:elementary-accessibility-bounds}, this proves
\eqref{eq:sharp-pearson-accessibility}.
\end{proof}

\end{appendix}

\section*{Acknowledgements}

This work is funded by the Sandia Laboratory Research and Development Program. 
We acknowledge use of GPT 5.5 and 5.6 to assist in drafting the manuscript text, structuring the mathematical derivations, and generating the simulation code.

This article has been authored by an employee of National Technology \& Engineering Solutions of Sandia, LLC under Contract No. DE-NA0003525 with the U.S. Department of Energy (DOE). The employee owns all right, title and interest in and to the article and is solely responsible for its contents. The United States Government retains and the publisher, by accepting the article for publication, acknowledges that the United States Government retains a non-exclusive, paid-up, irrevocable, world-wide license to publish or reproduce the published form of this article or allow others to do so, for United States Government purposes. The DOE will provide public access to these results of federally sponsored research in accordance with the DOE Public Access Plan
\smallskip
\begin{center}
    \verb|https://www.energy.gov/downloads/doe-public-access-plan|
\end{center}

\bibliographystyle{imsart-number}
\bibliography{references}

\end{document}